\documentclass[11pt,letterpaper]{article}
\usepackage[margin=1in]{geometry}
\usepackage{fullpage}
\usepackage{times}
\usepackage{float}
\usepackage{amsfonts,amsmath,amssymb,amsthm,graphicx}
\usepackage{url}
\usepackage{subfig}
\usepackage{hyperref}
\usepackage{multirow}
\usepackage{caption} 
\usepackage{comment}
\usepackage{enumitem}
\usepackage[capitalize]{cleveref}
\usepackage[numbers]{natbib}

\newcommand{\given}{\,|\,}

\newcommand{\prob}[2][]{\text{\bf Pr}\ifthenelse{\not\equal{}{#1}}{_{#1}}{}\!\left[{\def\givenn{\middle|}#2}\right]}
\newcommand{\expect}[2][]{\text{\bf E}\ifthenelse{\not\equal{}{#1}}{_{#1}}{}\!\left[{\def\givenn{\middle|}#2}\right]}

\newcommand{\tparen}{\big}
\newcommand{\tprob}[2][]{\text{\bf Pr}\ifthenelse{\not\equal{}{#1}}{_{#1}}{}\tparen[{\def\given{\tparen|}#2}\tparen]}
\newcommand{\texpect}[2][]{\text{\bf E}\ifthenelse{\not\equal{}{#1}}{_{#1}}{}\tparen[{\def\given{\tparen|}#2}\tparen]}

\newcommand{\sprob}[2][]{\text{\bf Pr}\ifthenelse{\not\equal{}{#1}}{_{#1}}{}[#2]}
\newcommand{\sexpect}[2][]{\text{\bf E}\ifthenelse{\not\equal{}{#1}}{_{#1}}{}[#2]}

\newcommand{\rbr}[1]{\left(\,#1\,\right)}

\newcommand{\reals}{{\mathbb R}}
\newcommand{\posreals}{\reals_+}

\newcommand{\opt}[1]{#1^{\star}}

\newcommand{\rev}{\mathrm{Rev}}

\newcommand{\OPT}{\mathrm{OPT}}
\newcommand{\E}{\mathbb{E}}
\newcommand{\R}{\mathbb{R}}
\newcommand{\indicator}[1]{\mathbbm{1}[#1]}

\newcommand{\VAL}{\mathrm{VAL}}

\usepackage{xcolor}

\usepackage{thm-restate}

\usepackage{booktabs} 

\usepackage{comment}
\usepackage[utf8]{inputenc}
\usepackage{lipsum}
\usepackage{xargs}   
\usepackage{bbm}
\usepackage{url}
\usepackage{nicefrac,xfrac}
\usepackage{enumitem}
\setenumerate{wide = 0pt, labelwidth = 1.3333em, labelsep = 0.3333em, leftmargin = \dimexpr\labelwidth + \labelsep\relax, noitemsep,topsep=3pt}
\usepackage[colorinlistoftodos,prependcaption,textsize=tiny]{todonotes}
\newcommandx{\improvement}[2][1=]{\todo[linecolor=Plum,backgroundcolor=Plum!25,bordercolor=Plum,#1]{#2}}
\usepackage{amsfonts}
\usepackage{graphicx}

\usepackage{algorithm}
\usepackage{algpseudocode}
\usepackage{mathtools,xparse}
\DeclarePairedDelimiter\ceil{\lceil}{\rceil}

\usepackage{mathrsfs}
\usepackage{mathtools}
\usepackage{lipsum}
\usepackage{hyperref}
\usepackage[capitalize]{cleveref}
\usepackage{enumitem}
\usepackage{todonotes}
\usepackage{xcolor}

\allowdisplaybreaks

\newtheorem{theorem}{Theorem}[section]
\newtheorem*{theorem*}{Theorem}
\newtheorem{lemma}[theorem]{Lemma}
\newtheorem*{lemma*}{Lemma}
\newtheorem{definition}[theorem]{Definition}

\newtheorem{claim}[theorem]{Claim}

\newtheorem{proposition}[theorem]{Proposition}

\newtheorem*{observation*}{Observation}

\newcommand{\nonl}{\renewcommand{\nl}{\let\nl\oldnl}}
\usepackage{graphicx}
\makeatletter
\newcommand*{\rom}[1]{\expandafter\@slowromancap\romannumeral #1@}
\makeatother

\newcommand{\tal}[1]{{\color{blue}[TA: #1]}}

\usepackage{color-edits}

\addauthor{yl}{red}

\setcitestyle{authoryear}

\DeclareMathOperator{\argmax}{\arg\max}

\title{Multi-Project Contracts}

\begin{document}


\author{
Tal Alon \\
\small{alontal@campus.technion.ac.il} \\
{Technion--Israel Institute of Technology} \\
\and
{Matteo Castiglioni} \\
\small{matteo.castiglioni@polimi.it} \\
{Politecnico di Milano} \\
\and 
{Junjie Chen} \\
\small{chen.junjie@inf.kyushu-u.ac.jp} \\
{Kyushu University} \\
\and
{Tomer Ezra} \\
\small{tomer@cmsa.fas.harvard.edu} \\ 
{Harvard University} \\
\and
{Yingkai Li} \\
\small{yk.li@nus.edu.sg} \\
{National University of Singapore} \\
\and
{Inbal Talgam-Cohen } \\
\small{inbaltalgam@gmail.com} \\
{Tel Aviv University} \\
}

\date{}

\begin{titlepage}
	\clearpage\maketitle
	\thispagestyle{empty}

\begin{abstract}

We study a new class of contract design problems where a principal delegates the execution of multiple projects to a set of agents. The principal's expected reward from each project is a combinatorial function of the agents working on it. Each agent has limited capacity and can work on at most one project, and the agents are heterogeneous, with different costs and contributions for participating in different projects. The main challenge of the principal is to decide how to allocate the agents to projects when the number of projects grows in scale. 

We analyze this problem under different assumptions on the structure of the expected reward functions. As our main result, for XOS functions we show how to derive a constant approximation to the optimal multi-project contract in polynomial time, given access to value and demand oracles. Along the way (and of possible independent interest), we develop approximate demand queries for \emph{capped} subadditive functions, by reducing to demand queries for the original functions. Our work paves the way to combinatorial contract design in richer settings. 
\end{abstract}

\end{titlepage}







\section{Introduction}

Contract theory addresses a fundamental question: \emph{how should we incentivize individuals to exert effort?} In practice, contracts structure 
market interactions, such as pricing agreements between cloud computing providers and clients, or revenue-sharing contracts between online platforms and sellers. 

The theory of contracts, recognized by the 2016 Nobel Prize, tackles {\it moral hazard} challenge arising in principal-agent settings. In the most basic setting,  a principal delegates the execution of a project to an agent, where the agent can (privately) choose a costly action to complete the project. The action taken by the agent leads to a stochastic outcome that is observable to the principal. To incentivize a desirable action, the principal designs a payment scheme (a.k.a.~contract) that pays according to the stochastic outcome. Such {\it performance-based} incentives are widely adopted in practice, with empirical evidence---such as studies on online labor markets---demonstrating its effectiveness \cite{kaynar2023estimating}. 
In recent years, contract theory is gaining attention from computer scientists \cite{babaioff2006combinatorial,HoSV16,dutting2020simple}, and a rapidly growing literature on algorithmic contract design has emerged. 
A key driver of the increasing interest is the shift from classic contracts to online digital contracts, with applications like online labor markets \cite{fest2020motivation,wang2022contract} or online content creation \cite{yao2024rethinking}. 
See \citet{dutting2024algorithmic} for a survey.

\vspace{1mm}
{\it The need for new models.} 
Large-scale, sophisticated contractual environments call for new models that capture their added complexity, as well as any structure that can be utilized to design approximately-optimal contracts despite this complexity. 
Such models complement classic economic ones, which are often more focused and apply assumptions that allow for clean, closed-form optimal solutions. 
To pave the road towards practical implementation of algorithmic contracts, increasingly general models are being introduced by the computational literature, which uses approximation to curb their complexity.
Two recent examples include models with a heterogeneous agent population~\cite{AlonDT21,alon2023bayesian,castiglioni2025reduction,castiglioni2022designing,guruganesh20,guruganesh2023power}, or with information structures in combination with contracts~\cite{castiglioni2025hiring,babichenko2022information,garrett2023optimal}.

We are particularly interested in \emph{combinatorial contracts}, the design of which requires finding a suitable combination---e.g., of agents to form a team, outcomes to serve as a performance measure, or actions to advance a project. 
Previous works have studied the combinatorial aspects of a single principal-agent pair \cite{DuttingEFK21,DuttingRT21,dutting2024combinatorial,deo2024supermodular,ezra2023inapproximabilitycombinatorialcontracts}, as well as combinatorial aspects that arise from forming a team of multiple agents~\cite{babaioff2006combinatorial,DuettingEFK23,duetting2024multiagent,CastiglioniEtAl23}.
Identifying new combinatorial dimensions of contracts is highly relevant to understanding their computational aspects and possible sources of complexity.

\vspace{1mm}
{\it Our model: Multi-project contracts.} 
We introduce {\it multi-project contracts}, a new class of combinatorial contract settings that further enriches the field of combinatorial contract design. In a multi-project setting, there is a single principal, $n$ agents and $m$ projects. The principal must match agents to projects, possibly leaving some agents/projects unallocated. Multiple agents can work on the same project, but agents have limited capacity and
each agent can be assigned to at most one project. Each project $j\in[m]$ has a \emph{success function} $f_j$, which maps possible sets of assigned agents to the project's success probability. We assume monotone non-decreasing success functions.


Each agent $i\in[n]$ has a cost $c_{ij}$ for working on project $j$---e.g., this cost might be higher for projects that require the agent to acquire new skills, or for projects that require tedious legwork. To incentivize work despite this cost, the principal designs a linear contract for each agent. Say a team $S$ is assigned to a project $j$, and consider an agent $i\in S$; then if the project succeeds, the agent is paid $\alpha_{ij}$ of the reward generated for the principal, where $\alpha_{ij}$ is the agent's cost $c_{ij}$ divided by her marginal contribution to $S$. 
We know from \cite{babaioff2006combinatorial} that given an assignment of agents to projects, the best way to incentivize the agents is by paying each agent the cost over marginals.
Thus, the main challenge is to decide on the allocation. 
We study this question for different classes of set functions to which the projects' success functions belong, with our main focus on XOS functions.  


The multi-project contract setting presents 
new opportunities and fresh challenges compared to previous models. 
Many practical scenarios fit into this new type of problem. For example, a company usually has many departments (i.e., projects), and employees are assigned to different departments according to their {\it abilities}. These abilities can be interpreted as the agents' contributions to the success of the different projects in our model, as well as their respective costs. 
Going from a single project to multiple projects is akin to going from selling items to a single buyer, to selling them to multiple buyers in auction design: 
in the former setting, a subset of items is chosen and allocated to the buyer, whereas in the latter setting, the items must be (partially) partitioned among the buyers.  
Similarly, when designing multi-project contracts, the principal faces two fundamental challenges: The principal needs to {\it simultaneously} determine i)~which subset of agents will be selected for the projects, and ii)~how to allocate this subset of agents to different projects. Single-project contract settings pose only the first challenge.





\subsection{Our Contributions and Techniques}
Our main result is a polynomial-time approximation algorithm for the multi-project contract problem with submodular and XOS reward functions.

\begin{theorem*}
When all projects have \emph{XOS} reward functions, there exists a polynomial-time algorithm that computes a contract that is an $O(1)$-approximation to the optimal contract using value and demand queries. For \emph{submodular} reward functions, the same guarantee holds using only value queries.
\end{theorem*}

Our approximation algorithm distinguishes between two types of projects in the optimal allocation: \emph{projects with dominant agents}, where at least one agent has a significant contribution to the project's revenue, and \emph{projects without dominant agents}, where all agents have relatively small contributions. For each type, we design a separate constant-approximation algorithm. Combining these two algorithms yields a constant approximation to the grand optimal allocation.

For projects with dominant agents, we construct a weighted complete bipartite graph, where one side represents agents and the other represents projects. We set each edge's weight $w_{i,j}$ to be the revenue generated by assigning agent $i$ to project $j$ (or zero if the revenue is negative). We then approximate the contribution of these projects by simply finding the maximum weighted matching in this graph. See \cref{alg:appx-opt-minus} for detailed constructions.

The main challenge of the multi-project contract problem is approximating the revenue of projects without dominant agents. 
The first natural attempt is to extend the approach of \cite{DuettingEFK23}, which applies to single-project case, to our multi-project setting. This approach would proceed as follows: 
Iterate over all estimates of the rewards of projects in the optimal allocation, denoted as $x=(x_1, \dots, x_m) \sim f(S^{\star}_1), \dots, f(S^{\star}_m)$, ensuring that at least one vector contains estimates within a constant factor of the true values.\footnote{An estimate $x_j$ can for example be $n$ discrete values increasing by factors of $2$, i.e., $2^i \max_i f_j(\{i\})$ for $i \in [\log n]$.} For each such vector of estimates, compute an allocation that maximizes \begin{equation}\label{eq:intro-esti}
    (S^x_1,\ldots,S^x_m)=S^x \in \arg\max_{S_1,\ldots,S_m} \left\{\sum_{j\in M} \left( f_j(S_j) - \sum_{i\in S_j} p_{i,x_j} \right) \right\},
    \end{equation}
where $p_{i,x_j} = \sqrt{\frac{x_j c_{ij}}{2}}$ for every agent $i$ and estimate $x_j$ for project $j$. Then, for each allocation $S^x$ and each project $j$, we scale $S_j^x$ by removing agents to reduce the reward by at most a constant factor while keeping the payments necessary to incentivize the allocated agents low. \citet{DuettingEFK23} demonstrate this for a single project, showing that with a close estimate of $x$ (i.e., near the optimal reward), scaling the set maximizing \eqref{eq:intro-esti} results in a high-reward allocation with low payments to the agent, achieving revenue close to the optimal. 

However, this approach fails in several steps along the way when considering multiple projects. First, going over all potential vectors of $x$ is intractable for two reasons. First is the exponential search space, since the number of possible options for $x$ grows exponentially in the number of projects. 
Second, while in the single project case,  maximizing \eqref{eq:intro-esti} can be done using a single demand query, in the multi-project case, it is equivalent to maximizing welfare given personalized prices which is more evolved, and requires new techniques. 

\vspace{1mm}
{\it LP formulation of the problem.}  To address these challenges, in \cref{sub:opt+-approx}, we formulate a linear program (LP) that defines a variable $y_{j,x,S}$ for each project $j$, each possible reward estimate $x$ and each set of agents $S$. 
The feasibility constraint of the LP is such that no agent is (fractionally) allocated to more than one project, and no project is (fractionally) allocated with more than one set of agents. We then rewrite \eqref{eq:intro-esti} using the new variables, capping $f_j(S_j)$ at its corresponding estimate $x_j$. This capping is crucial in our setting. Unlike \cite{DuettingEFK23}, which can iterate over estimates $x$, compute $S^x$ from \eqref{eq:intro-esti}, scale $S^x$, and select the best solution, our LP produces a solution corresponding to a single estimate. Scaling is effective when the estimate is close to the true function value, and capping prevents obtaining a reward that is too high to be properly scaled.

This LP has a polynomial number of constraints but exponentially many variables. In order to solve it in polynomial time, we use the ellipsoid method with an approximate separation oracle on its dual formulation. 

\vspace{1mm}
{\it Computing capped demand.} One of the key tools in our analysis, which may be of independent interest, is the construction of an approximate demand oracle for \emph{capped} subadditive function, assuming demand oracle access to the original function. See \cref{sub:approximate-demad}. 
Specifically, given a subadditive function $f:2^A \rightarrow \R_{\geq 0}$ and a parameter $x >0$, we define the capped function as $\hat{f}_x(S) = \min\{f(S),x\}$ for every $S\subseteq A$. 
Note that a demand oracle for the original function $f$ does not immediately extend $\hat{f}_x$. We provide a polynomial time algorithm for computing an approximate demand for $\hat{f}_x$. 
Exact demand computation is challenging---even when $f$ is additive, where demand queries are otherwise straightforward. The difficulty arises because $\hat{f}_x$, corresponding to an additive $f$, becomes \textit{budget-additive}, reducing exact demand computation to solving a \textit{knapsack problem}, which is NP-hard. 
Equipped with this key lemma, we can compute an approximate separation oracle for the dual LP and use it to solve the problem efficiently via the ellipsoid method. This allows us to find an approximately optimal fractional solution to the LP, ensuring that the solution has a polynomial sizes support.

\vspace{1mm}
{\it Rounding procedure for XOS.} After solving the LP, we face another challenge, that does not arise in the single-project problem. The solution of the LP is fractional and may not correspond to a feasible integral allocation. To address this, our next step is to round the fractional solution to obtain an integral allocation of agents to projects. We provide a simple and intuitive rounding technique for XOS functions (see \cref{alg:rounding}). Our rounding ensures that (1) the integral allocation is feasible; and (2) the rounded value is a $2$-approximation to the fractional value.  

\vspace{1mm}
{\it Generalized scaling.}
Lastly, in \cref{sec:scaling}, we generalize the scaling lemma of \citet{DuettingEFK23} to work for our setting, and apply it to the rounded allocation to reduce the fraction of payments while maintaining a constant fraction of the reward. This final step ensures that the expected revenue from the scaled allocation provides a constant approximation to the optimal contribution for projects without dominant agents.

\subsection{Additional Related Work}
\label{appx:related-work}

\paragraph{(Algorithmic) Contract Design.} Over several decades, a vast literature on contract theory has been  developed \cite{bolton2004contract,carroll2015robustness,bajari2001incentives,gottlieb2015simple,chade2019disentangling}. Recently, there has been growing interest within the computer science community in studying the computational aspects of contract design. For example, \citet{dutting2020simple} study the efficiency of contracts with specific interests in the gap between optimal linear contracts and optimal contracts in non-Bayesian settings. A rich body of literature extends contract design problem to Bayesian settings. Specifically, \citet{AlonDT21, alon2023bayesian} study single-parameter Bayesian contract design where the agent’s type is represented by a single parameter, while \citet{guruganesh20, castiglioni2022designing, CastiglioniMG22, guruganesh2023power} study Bayesian contract design where agent’s type is multi-parameter, i.e., arbitrarily defining costs and action-outcome mapping matrices. \citet{castiglioni2025reduction} draw a connection between these two settings providing a reduction from multi-parameter to single-parameter Bayesian contract design. 
Other extensions studied by recent works include agent-designed contracts with a typed principal~\citet{bernasconi2024agent},  ambiguous contracts~\citep{dutting2024ambiguous}, learning  contracts~\cite{chen2024bounded, han2024learning,zhu2023sample}, and scoring rules~\cite{hartline2022optimization,hartline2023optimal}. 
We refer interested readers to the recent survey by \citet{dutting2024algorithmic}. In combinatorial contracts, existing works close to ours are \cite{DuettingEFK23,ezra2023inapproximabilitycombinatorialcontracts}. We remark that their negative results  carry over to our settings. Specifically, the negative result against the more general subadditive functions of \cite{DuettingEFK23} rules out the possibility of constant approximation algorithms. Regarding lower bounds, \citet{ezra2023inapproximabilitycombinatorialcontracts} show no polynomial time algorithm can guarantee better than 
$1.42$-approximation for submodular functions with value query access, and \citet{DuettingEFK23} show that one cannot guarantee better than 
$1.136$-approximation for XOS functions with value and demand query access. Additionally, while we focus on the complement-free hierarchy settings, \citet{deo2024supermodular} present a strong impossibility result for the single-project settings when the reward function is supermodular.


\paragraph{Welfare Maximization LP formulation.}  
As discussed in the main body of the paper, our formulation of \eqref{lp:lp1} resembles the LP formulation used for the welfare maximization problems. Similar to our setting, the solution to this LP requires a rounding step at the end. Here, we elaborate on the guarantees of different known rounding techniques. \cite{welfare-maximizing-uri-feige} study welfare maximization for subadditive utility functions using demand queries to solve the LP. They provide a randomized solution that achieves a $\frac{1}{2}$-approximation for subadditive functions and a $\left(1 - \frac{1}{e}\right)$-approximation for XOS functions. \cite{Feige2006TheAP} improve upon this by giving a $\left(1 - \frac{1}{e} + \epsilon\right)$-approximation for submodular utility functions for some constant $\epsilon$. The maximum submodular welfare problem is known to be APX-hard, even when each player has a utility function of constant size. \cite{Dobzinski-Schapira06}, which predates the work of Uri Feige, also provide a $\left(1 - \frac{1}{e}\right)$-approximation for XOS functions. \cite{lehmann2-nisan} show that a simple greedy algorithm, which allocates based on the highest marginal value, yields a $\frac{1}{2}$-approximation for welfare maximization using value queries. \cite{Dobzinski-nisan-schapira} were the first to formulate the welfare maximization problem as an LP with exponentially many variables but only polynomially many constraints. They derive the dual of this LP and demonstrate how to implement the separation oracle using a demand oracle. \cite{Vondrak-08} provide a $\left(1 - \frac{1}{e}\right)$-approximation to the optimal solution.  

\section{Model}

\paragraph{Setting.} We consider a {\it multi-project contract design} model, in which a single principal assigns agents to projects.
Let $N$ denote the set of agents and $M$ the set of projects, where $|N|=n$ and $|M|=m$.
All projects have the same set of binary outcomes, $\Omega = \{0, 1\}$, where \(\omega = 1\) represents \emph{project success} and \(\omega = 0\) represents \emph{project failure}.
The principal receives a reward of 1 for each successful project and a reward of 0 for each failed project.\footnote{It is without loss of generality to assume that the reward for success is the same across all projects, as we can always normalize the rewards by scaling the success probabilities, as defined below.} 
Each agent $i\in N$ can work on at most {\it one} of the projects. The cost agent $i$ incurs from working on project $j \in M$ is $c_{ij}\geq 0$, and the cost of exerting no effort is normalized to zero. 
For any project $j$, its success probability given a set of agents working on it is represented by a combinatorial function $f_j: 2^N \to [0, 1]$, known as the project's \emph{success function}. 
We assume that the success functions are normalized and monotone, i.e., for every project~$j$, $f_j(\emptyset)=0$ and $f_j(S) \leq f_j(S')$ for any sets of agents $S\subseteq S'\subseteq N$. 

\vspace{1mm}
{\it Contracts.} In contract design environments, the principal observes only the realized outcomes and not the agents' actual effort choices. Therefore, the principal must design payment schemes (i.e., contracts) to incentivize agents to exert costly effort and achieve desirable outcomes. Since each agent can exert effort in at most one project, and each project has a binary outcome and the reward for failure is zero, it is without loss of generality to focus on linear contracts with a single positive payment for each agent (or not paying the agent at all). In the context of multi-project contracts, a linear contract is defined as follows:  

\begin{definition}[Linear Contracts]
\label{def:linear}
In the multi-project contract design model, a \emph{linear} contract is represented by a matrix $t$, 
where for every $i\in N,j\in M$, $t_{ij}\in \reals_{\geq 0}$ is the payment to agent $i$ when project $j$ is successful.
A linear contract $t$ is \emph{single-payment} if for every agent $i\in N$, there exists at most one project $j\in M$ such that $t_{ij} > 0$.
\end{definition}

Note that a single-payment linear contract encodes an \emph{allocation} (partial partition) of the agents to the projects. Since the contracts we focus on are linear, we also use the notation $\alpha$ for a contract, as common in the literature.
%

For each project $j\in M$, to incentivize a set of agents $S_j$ to work on project $j$ in equilibrium, the following incentive constraint must hold for every agent $i\in S_j$:
\begin{align*}
f_j(S_j)  \cdot t_{ij} - c_{ij} \geq f_j(S_j \setminus \{i\})  \cdot t_{ij}.
\end{align*}
That is, conditioned on all agents $S_j \setminus \{i\}$ exerting effort under the contract, agent $i$ finds it beneficial to exert effort as well.\footnote{There may be additional equilibria, but in this paper, we assume that the principal can select the realized equilibrium.} 
This implies that fixing a set $S_j$, the optimal payment for incentivizing each agent $i\in S_j$ to exert effort satisfies 
\begin{align}
t_{ij} = t_{ij}(S_j) = \frac{c_{ij}}{f_j(S_j)-f_j(S_j\setminus \{i\})}.\label{eq:payments}
\end{align}
Moreover, $t_{ij}(S_j)=0$ for all $i\not\in S_j$. 

Given any allocation of agents to projects $\{S_j\}_{j\in [M]}$, 
the expected revenue (or simply, \emph{revenue}) of the principal from project $j$ is 
\begin{align*}
\rev_j(S_j) = \rbr{1-\sum_{i\in S_j} t_{ij}(S_j)}\cdot f_j(S_j)
= \rbr{ 1-\sum_{i \in S_j} \frac{c_{ij}}{ f_j(S_j) - f_j(S_j \setminus \{i\})} } \cdot f_j(S_j).
\end{align*}
Therefore, the optimization problem of the principal can be formulated as follows:
\begin{align*}
\max_{\{S_j\}_{j\in [M]}} \quad & \sum_{j\in M} \rev_j(S_j) \\
\textnormal{s.t.}\qquad & S_j \cap S_{j'} = \emptyset, \quad\forall j,j' \in M.
\end{align*}
We conclude that solving the multi-project contract design problem boils down to finding the optimal allocation of agents to projects; the contractual payments then follow from Equation~\eqref{eq:payments}.

\vspace{1mm}
{\it Classes of success functions.} In this paper, we consider success functions that belong to one of the following increasingly-general classes of set functions: additive, submodular, XOS, and subadditive.
\begin{itemize}
\item {\it Additive functions}. The function $f$ is additive if there exist $v_1, v_2, \dots, v_n \in \reals_+$ such that $f(S) = \sum_{i \in S} v_i$ for any $S\subseteq N$ .
\item {\it Submodular functions}. The function $f$ is submodular if for any two sets $S \subseteq S'\subseteq N$ and any $i \in N$, it holds that $f(S\cup\{i\}) - f(S) \ge f(S'\cup\{i\}) - f(S')$.  
\item {\it XOS functions}. The function $f$ is XOS if there exists a collection of additive set functions $\{a_1, a_2, \dots, a_l\}$ such that for any set $S\subseteq N$, it holds that $f(S) = \max_{i\in [l]} a_i(S)$.

\item {\it Subadditive functions}. The function $f$ is subadditive if for any two subsets $S, T \subseteq N$, it holds that $f(S) + f(T) \ge  f(S\cup T)$.
\end{itemize}

A standard notation is $f(S\mid S') := f(S\cup S')-f(S')$ for any two sets $S,S'\subseteq N$,  and $ f( i \mid S) = f(S\cup\{i\} )-f(S) $.

\vspace{1mm}
{\it Oracle access.} Following the literature on optimization with combinatorial set functions, we employ two standard oracles:
\begin{itemize}
    \item A \emph{value oracle} for accessing $f$ returns the value of $f(S)$ given any input set $S\subseteq N$.
    \item A \emph{demand oracle} for accessing $f$ returns a set $S \subseteq N$ that maximizes $f(S) - \sum_{i\in S}p_i$ given any price vector $p=(p_1, p_2,\dots, p_n) \in \reals^n_+$.
\end{itemize}
Both value and demand oracles are employed and have been shown to be useful in previous contract design problems involving combinatorial optimization \citep[e.g.,][]{DuettingEFK23,DuttingEFK21,duetting2024multiagent}. 


\section{Main Theorem and Proof Outline} 
\label{sec:const-appx-diff-tasks}
In this section, we present our main result --- a polynomial-time algorithm that computes an approximately optimal allocation when the success functions are XOS. We then provide a high-level outline of the proof, with full details given in the subsequent sections.

\begin{theorem}\label{thm:const-diff-projects}
When all projects have XOS reward functions, there exists a polynomial-time algorithm that computes a contract that is an $O(1)$-approximation to the optimal contract using value and demand queries. For submodular reward functions, the same guarantee holds using only value queries.
\end{theorem}

\noindent
We prove this theorem for XOS functions and defer the proof for submodular success functions to Appendix~\ref{appx:sub-mod}. 
We first introduce definitions and establish key claims.  

Given a multi-project instance, we denote the optimal contract $\opt{\alpha} = \{\opt{\alpha}_j\}_{j\in M}$ and the corresponding optimal allocation of agents to projects as $\opt{S}=\{\opt{S}_j\}_{j\in M}$. Given allocation $S=\{S_j\}_{j\in M}$,
let $\rev(S) =\sum_{j\in M} \rev_j(S_j)$ be the expected revenue when the principal sets the optimal linear contract for $S$ according to \Cref{eq:payments},
and let $\OPT=\rev(\opt{S})$ be the optimal revenue. 
We say an agent $i$ is a \emph{dominant} agent for project $j$ if its contribution to project $j$ is high compared to the optimal reward of this project, i.e., $f_j(\{i\}) > \delta  f_j(\opt{S}_j)$ for a small constant $\delta$ (defined below), and is a \emph{nondominant} agent otherwise.
Let
\begin{eqnarray*}
\delta &:=& 1/129\\
\opt{A}_j&:=&\{ i \mid f_j(\{i\})\leq \delta f_j(S_j^{\star}) \},\\
I^{\star}&:=&\{j\mid S^{\star}_j \subseteq A^{\star}_j\}.
\end{eqnarray*}
Intuitively, $\opt{A}_j$ is the set of \emph{nondominant} agents. $I^{\star}$ is the set of projects with no dominant agent according to the optimal allocation $\opt{S}$.
Next, define
\begin{eqnarray*}
    \OPT_{-}&:=& \sum_{j\not \in I^{\star}}\rev_j(\opt{S}_j),\\
    \OPT_{+}&:=& \sum_{j \in I^{\star}}\rev_j(\opt{S}_j),
\end{eqnarray*}
where $\OPT_{-}$ represents the portion of $\OPT$ contributed by projects with dominant agents, and $\OPT_{+}$ corresponds to the portion from projects without dominant agents.

We present two polynomial-time algorithms for computing contracts that achieve constant-factor approximations to $\OPT_-$ and $\OPT_+$ separately.

\paragraph{Proof outline for $\OPT_-$ approximation.} This is the less involved case. We leverage the fact that for each project $j \not\in I^{\star}$, there exists a dominant agent $i_j$ in the optimal allocation who contributes at least a $\delta$ fraction of $f_j(S^{\star}_j)$. By subadditivity, removing all agents except $i_j$ ensures that $f_j(\{i_j\}) \geq \delta f_j(S^{\star}_j)$. Therefore, assigning only this dominant agent to each project $j \not\in I^{\star}$ achieves at least $\delta$ fraction of $\OPT_-$. 
This implies that achieving a constant-factor approximation to $\OPT_-$ reduces to finding the best allocation that incentivizes at most one agent per project. This constrained optimization problem can be formulated as a maximum bipartite matching problem, which can be solved in polynomial time using \cref{alg:appx-opt-minus} (see \cref{sub:opt_-approx}).

\paragraph{Proof outline for $\OPT_+$ approximation.}  
We begin by formulating a linear program \eqref{lp:lp1} that computes a fractional allocation of agents to projects. This LP resembles linear programs used in welfare maximization problems (e.g., \cite{welfare-maximizing-uri-feige}), where the variables represent fractional relaxation of indicator functions for allocating sets of agents to projects, and the constraints ensure the feasibility of the fractional solutions. 
Note that while \eqref{lp:lp1} has a polynomial number of constraints, it contains exponentially many variables. To solve it in polynomial time and obtain a fractional (or equivalently, randomized) allocation $\tilde{S}$ that is approximately optimal for \eqref{lp:lp1}, we derive its dual \eqref{LP:lp2} and design a polynomial-time approximate separation oracle (see Claim~\ref{claim:appx-separation}). 
A key component of this oracle, which may be of independent interest, is an \emph{approximate} demand oracle for capped sub-additive functions (see \cref{sub:approximate-demad} and \cref{alg:appx-demand}). Furthermore, in \cref{alg:solve-lp1}, using the constructed approximate separation oracle, we compute an approximately optimal fractional solution to \eqref{lp:lp1} via the ellipsoid method. 
By our choice of objective function for \eqref{lp:lp1}, the expected reward of this randomized allocation $ \mathbb{E}[\sum_{j\in M} f_j(\tilde{S}_j)]$ is close to the reward of the optimal allocation, i.e.,  
$ \mathbb{E}[\sum_{j\in M} f_j(\tilde{S}_j)] = \Theta(1) \cdot \sum_{j\in I^{\star}} f(S^{\star}_j), $  
while ensuring that the agents' marginal contributions to the projects are sufficiently high. Note that since payments to the agents are determined by their marginal contributions (see \eqref{eq:payments}), high marginals imply lower payments from the principal to the agent, which is desirable for achieving a higher revenue. 
Next, we apply a deterministic rounding algorithm (\cref{alg:rounding}) to round this fractional solution to obtain a feasible integral allocation $T$ that is close to the expected reward of the fractional solution, i.e.,  
$ \sum_{j\in M}f(T) = \Theta(1)\cdot \mathbb{E}[\sum_{j\in M}f_j(\tilde{S}_j)] = \Theta(1)\cdot \sum_{j\in I^{\star}}f(S^{\star}_j) $  
(see \cref{lemma:rounding}), while maintaining high marginal contributions from the agents. As a final step, we generalize the scaling lemma of \citet{DuettingEFK23}, which allows us to scale down the value $f_j(T_j)$ for each project $j$ by any constant factor while preserving high marginal contributions from the agents. That is, $f_j(i \mid T_j\setminus \{i\}) \geq \sqrt{2f_j(T_j)}$ for any $i\in T_j$. 
These high marginal contributions preserved in all previous operations enable us to apply Lemma 3.4 from \cite{duetting2024multiagent} (see \cref{lemma:lemma34-multi}) to obtain that
$ \rev(T) \geq 0.5 \sum_{j\in M}f_j(T_j) = \Theta(1)\cdot \sum_{j\in I^{\star}}f_j(S^{\star}_j). $  
Since payments are non-negative, this establishes a constant-factor approximation to $\OPT_+$.  

Finally, \cref{alg:appx-opt} computes the contracts that are approximately optimal for $\OPT_-$ and $\OPT_+$ respectively. Comparing those contracts and selecting the one with the higher revenue yields a constant-factor approximation to the optimal revenue $\OPT$.  

\section{Projects with a Dominant Agent}
\label{sub:opt_-approx}

In this section, we provide a polynomial-time algorithm (\cref{alg:appx-opt-minus}) to approximate $\OPT_{-}$, the revenue contribution from projects that include dominant agents in the optimal allocation.

\begin{lemma}[Approximation of $\OPT_{-}$]\label{lem:opt-minus}
For multi-project contract settings with subadditive success functions, Algorithm~\ref{alg:appx-opt-minus} computes an allocation $S^{-}$
with $\rev(S^{-})\geq \delta \cdot  \OPT_{-}$ in polynomial time, using a polynomial number of value queries.
\end{lemma}

\begin{proof}
We first show that there \emph{exists} an allocation $S$ such that $|S_j|\leq 1$ for any $j\in M$ and its revenue is $\rev(S)\geq \delta \cdot \OPT_{-}$. The existence is constructed as follows. Consider the following allocation $S$ such that for each project $j$, the allocation $S_j$ is obtained by allocating a dominant agent (if exists) with the highest individual contribution among the agents in the optimal allocation:
\begin{eqnarray*}
    S_j = \begin{cases}
        \arg\max_{i_j\in S_j^\star} \{ f_j(\{i_j\}) \mid i_j \not \in \opt{A}_j \} & j\not \in I^{\star},\\
        \emptyset & j\in I^{\star}.
    \end{cases}
\end{eqnarray*}
Let $i_j$ be the agent in $S_j$ if $j\notin I^{\star}$. By subadditivity and since $f_j(\emptyset)=0$, 
we have $f_j(i_j\mid \emptyset)= f_j(\{i_j\})\geq  f_j(i_j\mid \opt{S}_j \setminus \{i_j\}).$ Therefore, 
\begin{eqnarray}\label{eq:rev-opt-minus}
    \rev(S)=\sum_{j\not \in I^{\star}} f_j(\{i_j\})\left(1-\frac{c_{ij}}{f_j(i_j\mid \emptyset)}\right)\geq \sum_{j\not \in I^{\star}} f_j(\{i_j\})\left(1-\frac{c_{ij}}{f(i_j\mid \opt{S_j}\setminus \{i_j\})}\right). \nonumber
\end{eqnarray}
For every project $j\in M$, monotonicity of $f_j$ implies that $f_j(i\mid \opt{S}_j \setminus \{i\})\geq 0$ for all $i\in S^{\star}_j$. Therefore, 
\begin{eqnarray*}
\rev(S)\geq \sum_{j\not \in I^{\star}} f_j(\{i_j\}) \rbr{1-\sum_{i\in \opt{S_j}}\frac{c_{ij}}{f(i\mid \opt{S_j}\setminus \{i\})}} \geq \sum_{j\not \in I^{\star}} \delta   f_j(S_j^\star) \rbr{1-\sum_{i\in \opt{S_j}}\frac{c_{ij}}{f(i\mid \opt{S_j}\setminus \{i\})}} = \delta  \OPT_{-},
\end{eqnarray*}
where the first inequality holds since we are subtracting additional non-negative payments and the second inequality holds since $i_j$ is a dominant agent for project $j$ with $f_j(\{i_j\})\ge \delta \cdot f_j(S_j^\star$).

Even though we do not know how to compute the set $S$ constructed in the former paragraph  about existence, 
our Algorithm~\ref{alg:dom-agent} computes an allocation $S^{-}$ that satisfies $|S^{-}_j|\leq 1$ for any $j\in M$ and $ \rev(S^-) \geq \rev(S)$, using a polynomial number of value queries. 
In particular, Algorithm~\ref{alg:dom-agent} finds a maximum weighted matching, where a valid matching in this graph corresponds to an allocation of agents to projects such that each project receives at most one agent. Furthermore, by the definition of the weights, the total weight of a valid matching equals the revenue of the corresponding allocation. Since $S$ is also a valid matching, and $S^-$ is a maximum weighted matching under the same weight function, it follows that $\rev(S^-) \geq \rev(S) \geq \delta \cdot \OPT_-$.
\end{proof}

\begin{algorithm}[t]
\caption{Constant-Factor Approximation for $\OPT_-$\\
\textbf{Input:} Set of agents $N$, project success functions $f_1, \dots, f_m$, and agent costs $c_{ij}$ for every $i \in N,j \in M$\\
\textbf{Output:} Allocation $S^-$ that approximates $\OPT^-$
}\label{alg:appx-opt-minus}
\begin{algorithmic}[1]
\State Construct a bipartite graph $G = (N, M, w)$ where $N$ represents agents, $M$ represents projects, and the weight of each edge $(i,j)\in N\times M$ is $w(i,j) = \max\{f_j(\{i\}) - c_{ij}, 0\}$
\State Compute $S^-$ as the maximum weighted matching in $G$.
\State \Return $S^{-}$ 
\end{algorithmic}
\label{alg:dom-agent}
\end{algorithm}

\section{Key Tool: Approximate Demand Oracle for Capped Subadditive Functions}\label{sub:approximate-demad}

In this section, we present a key technical lemma that is central to our polynomial-time algorithm for approximating $\OPT_+$. Specifically, we construct an approximate demand oracle for a capped subadditive function, assuming access to a demand oracle for the original (uncapped) function. 
 
\begin{lemma}[Approximate demand for capped function]
\label{lemma:appx-demand}
Let $f:2^A \to \mathbb{R}_+$ be a monotone subadditive set function over a family $A$ of $n$ elements.
Consider a cap $x \in \posreals$ and a parameter $\delta\in (0,1)$ such that $f(\{i\})\leq \delta \cdot  x$ for every $i\in A$. Given a price $p_i \geq 0$ for every element $i\in A$, Algorithm~\ref{alg:appx-demand} computes a set $\bar S$ in polynomial time in $n$ using demand queries to $f$, such that $\bar S$ satisfies the following
\begin{eqnarray}\label{eq:appx-demand}
    \min\{f(\bar{S}),x 
    \} - \sum\nolimits_{i \in \bar{S}} p_i \geq \frac{1}{1+\frac{1}{1-\delta}} 
    \max_{S} \left(\min\{f(S), x 
    \} - \sum\nolimits_{i \in S} p_i \right) - \delta x.
\end{eqnarray}
\end{lemma}

\begin{algorithm}[t]
\caption{Constant-Factor Approximation for Capped Demand\\
\textbf{Input:} Subadditive function $f$, capping value $x > 0$, parameter $\delta \in [0, 1]$ s.t. $f(i) \leq \delta  x$ $\forall i$, and payments $(p_i)_{i\in [n]}$.\\
\textbf{Output:} A set $\bar{S}$ satisfying Inequality~\eqref{eq:appx-demand}.}
\label{alg:appx-demand}
\begin{algorithmic}[1]
\State Initialize $S_1\gets \emptyset$ and $S_2 \in \arg \max_S \{f(S)- \sum\nolimits_{i \in S} p_i\}$.
\If{$f(S_2)\le x$}
\State \Return $S_2$\label{state:demand-orig}
\EndIf
\State $\gamma_1\gets 0$, $\gamma_2 \gets 1$ 
\While {$\gamma_2-\gamma_1\ge \delta$} \label{state:start-binary-1}
\State $\bar \gamma\gets \frac{\gamma_1+\gamma_2}{2}$ 
\State Let $\bar S$ be an arbitrary set maximizing $f(S)- \sum\nolimits_{i \in S} p_{i}/\bar \gamma $
\If {$f(\bar S)> x$}
    \State $S_2\gets \bar S$; $\gamma_2\gets \bar \gamma$.
\Else
    \State $S_1\gets \bar S$; $\gamma_1\gets \bar \gamma$.
\EndIf
\EndWhile \label{state:end-binary}

\State Initialize $j=1$ and \( U_j = \emptyset \).\label{state:partition}

\For{$i\in S_2$}

\State $U_j \gets U_j \cup \{i\}$.
\If{$\sum\nolimits_{j=1}^\ell f(U_j) \geq f(S_2)$}
\State jump  line \ref{get_ustar}
\EndIf
    \If{$f(U_j) \geq (1-\delta)x$}
        \State $ j \gets j + 1 $ and initialize $ U_j = \emptyset$.
    \EndIf
\EndFor

\State $U^{\star} \gets \arg \max_{j\in [\ell]}\{f(U_j)-\sum\nolimits_{i\in U_j}p_i\}$ \label{get_ustar}

\State $\bar{S} \gets S\in \arg\max_{S \in \{S_1, U^{\star}\}} \{ f(S) - \sum\nolimits_{i \in S} p_i \}$

\State \Return $\bar{S}$

\end{algorithmic}
\end{algorithm}

In Lemma~\ref{lemma:appx-demand}, the smaller $\delta$ is, the better the approximation guarantee. 
Towards proving Lemma~\ref{lemma:appx-demand}, we start by introducing several definitions. Let 
$$S^{\star}\in \arg \max_S \rbr{\min\{f(S),x\} - \sum\nolimits_{i \in S} p_i},$$ 
be an optimal demand for the capped function. By subadditivity of $f$, we can assume without loss of generality that $f(S^{\star})\le (1+\delta) x$ and removing any item in the set leads to a value below $x$. Such a solution always exists since otherwise, we can remove agents from $S^{\star}$ maintaining that $f(S^{\star})\ge x$ and the prices paid are weakly lower. 

We next analyze  Algorithm~\ref{alg:appx-demand}.
\begin{claim}\label{claim:binary-search}
If \cref{alg:appx-demand} does not terminate at \cref{state:demand-orig}, at the end of its line~\ref{state:end-binary}, it holds that:
\begin{enumerate}
    \item $f(S_1)\le x< f(S_2)$, $S_j \in \arg \max_S (f(S)-\sum\nolimits_{i \in S} p_i/\gamma_j)$ for $j\in \{1,2\}$, and $\gamma_2-\gamma_1 \in (0,\delta)$.\label{item:obs-alg-claim}
    \item $f(S_1)\leq f(S^{\star})$. \label{item:s1-alg-claim}
    \item $(\mu\gamma_1 +(1-\mu)\gamma_2) f(S^{\star}) -\sum\nolimits_{i \in S^{\star}}p_i \le \mu (\gamma_2 f(S_1) -\sum\nolimits_{i \in S_1}p_i) + (1-\mu) (\gamma_2 f(S_2) -\sum\nolimits_{i \in S_2}p_i)$ $\forall \mu\in [0,1]$\label{item:ineq-mu-alg-claim}
\end{enumerate}
\end{claim}

\begin{proof}
    Since the binary search at line \ref{state:start-binary-1}  starts with  $\gamma_1=0$ such that $f(S_1)\le x$ and $\gamma_2$ such that $f(S_2)>x$, it stops at line~\ref{state:end-binary} outputting $\gamma_1,S_1,\gamma_2,S_2$ that satisfy \eqref{item:obs-alg-claim}.

We next show that $f(S_1)\leq f(S^{\star})$. By definition of $S^{\star}$ and $S_1$, we have
\begin{eqnarray*}
f(S^{\star})-\sum_{i\in S^{\star}}p_i\geq \min\{f(S^{\star}),x\}-\sum_{i\in S^{\star}}p_i \geq \min \{f(S_1),x\}-\sum_{i\in S_1} p_i = f(S_1)-\sum_{i\in S_1} p_i.    
\end{eqnarray*}
Therefore, $f(S^{\star})- \sum_{i \in S^{\star}} p_i-(f(S_1)-\sum_{i \in S_1} p_i)\geq 0.$ Adding $(\gamma_1-1) (f(S^{\star})-f(S_1))$ to both sides of the inequality, we have
\begin{eqnarray*}
     &(\gamma_1-1) (f(S^{\star})-f(S_1))\le  \gamma_1 f(S^{\star})- \sum_{i \in S^{\star}}p_i-(\gamma_1 f(S_1)  -\sum_{i \in S_1} p_i) \leq 0,
\end{eqnarray*}
where the last inequality is by optimality of $S_1$ with respect to payments $p_i/\gamma_1$. This implies that $f(S^{\star})-f(S_1)\ge0$ since $\gamma_1<1$. This completes the proof for \eqref{item:s1-alg-claim}.

We now prove the inequality in \eqref{item:ineq-mu-alg-claim}. By optimality of $S_1$ we have $\gamma_1 f(S_1) -\sum_{i \in S_1}p_i\ge \gamma_{1} f(S^{\star})-\sum_{i \in S^{\star}} p_i$, and multiplying it by $\mu$ we have $\mu(\gamma_1 f(S_1) -\sum_{i \in S_1}p_i)\ge \mu(\gamma_{1} f(S^{\star})-\sum_{i \in S^{\star}} p_i)$. The same holds for $\gamma_2$, $f(S_2)$, and factor $(1-\mu)$. That is, $(1-\mu)(\gamma_2 f(S_2) -\sum_{i \in S_1}p_i)\ge (1-\mu)(\gamma_{2} f(S^{\star})-\sum_{i \in S^{\star}} p_i)$. Combining these inequalities and rearranging, we have 
\begin{eqnarray*}
(\mu \gamma_1+(1-\mu) \gamma_2) f(S^{\star}) -\sum_{i \in S^{\star}} p_i  & \le &  \mu (\gamma_1 f(S_1) -\sum_{i \in S_1}p_i) + (1-\mu) (\gamma_2 f(S_2) -\sum_{i \in S_2}p_i) 
\\ & \leq & \mu (\gamma_2 f(S_1) -\sum_{i \in S_1}p_i) + (1-\mu) (\gamma_2 f(S_2) -\sum_{i \in S_2}p_i), 
\end{eqnarray*}
where the second inequality holds since $\gamma_2 >\gamma_1$.
This completes the proof.
\end{proof}

 We next prove the main lemma of this section.

\begin{proof}[Proof of Lemma~\ref{lemma:appx-demand}]
If the algorithm returns $\bar{S}=S_2$ at line \ref{state:demand-orig},
$$ \min\{f(\bar{S}),x\}-\sum\nolimits_{i \in \bar{S}} p_i= f(\bar{S})-\sum\nolimits_{i \in \bar{S}} p_i\ge \max_S (f(S)-\sum\nolimits_{i \in S} p_i) \geq \max_S (\min\{f(S),x\}-\sum\nolimits_{i \in S} p_i).$$
Otherwise, the binary search proceeds and ends satisfying the conditions in Claim \ref{claim:binary-search}. Recall that $f(S^{\star})\le (1+\delta)x$. 
Also, observe that $\frac{f(S_2)-x}{f(S_2)-f(S_1)}\cdot  f(S_1)+ \left(1- \frac{f(S_2)-x}{f(S_2)-f(S_1)}\right)\cdot f(S_2)=x$. Fixing 
\begin{eqnarray}\label{eq:mu}
\mu= \frac{f(S_2)-x}{f(S_2)-f(S_1)}, 
\end{eqnarray} in Claim~\ref{claim:binary-search}, it therefore holds that
$
f(S^{\star})\le  (1+\delta)x= (1+\delta)(\mu f(S_1)+ (1-\mu)f(S_2)).
$
Multiplying the inequality above by $(1-\gamma_2)$ and adding it to the inequality in Claim~\ref{claim:binary-search} \eqref{item:ineq-mu-alg-claim}, we have that
\begin{align*}
&(1-\gamma_2)f(S^{\star})+(\mu\gamma_1 +(1-\mu)\gamma_2) f(S^{\star}) -\sum\nolimits_{i \in S^{\star}}p_i \\
&\le (1-\gamma_2)(1+\delta)(\mu f(S_1)+ (1-\mu)f(S_2))+\mu (\gamma_2 f(S_1) -\sum\nolimits_{i \in S_1}p_i) + (1-\mu) (\gamma_2 f(S_2) -\sum\nolimits_{i \in S_2}p_i).
\end{align*}
After rearranging we get
\begin{align}\label{eq:s-star}
&f(S^{\star}) -\sum\nolimits_{i \in S^{\star}}p_i+\mu(\gamma_1 -\gamma_2) f(S^{\star}) \notag \\
&\le \mu [(1+\delta-\gamma_2\delta ) f(S_1)-\sum\nolimits_{i \in S_1}p_i]+(1-\mu)[ (1+\delta-\gamma_2\delta)f(S_2)   - \sum\nolimits_{i \in S_2}p_i] \nonumber \\
&=\mu (f(S_1)-\sum\nolimits_{i \in S_1}p_i)+(1-\mu) (f(S_2)   - \sum\nolimits_{i \in S_2}p_i)+\delta(1-\gamma_2 ) (\mu f(S_1)+(1-\mu)f(S_2)) \notag\\
&=\mu (f(S_1)-\sum\nolimits_{i \in S_1}p_i)+(1-\mu) (f(S_2)   - \sum\nolimits_{i \in S_2}p_i)+\delta(1-\gamma_2 ) x,
\end{align}
where the last inequality follows from the definition of $\mu$.

After establishing that the inequality holds at the end of the binary search, i.e., at line~\ref{state:end-binary}, we analyze the subsequent steps of the algorithm, beginning from line~\ref{state:partition}. The algorithm initializes an empty set $U_1$ and iteratively adds elements from $S_2$ to $U_j$ until the  value $f(U_j)$ falls within the interval $[(1-\delta)x, x]$. Once this condition is met, a new set is initialized, and this process continues until the cumulative function value reaches or exceeds $f(S_2)$.

We next show that the algorithm constructs at most $\ell \leq \left\lceil \frac{f(S_2)}{(1-\delta)x} \right\rceil$ 
sets, ensuring that each set $U_j$ with $j<\ell$ satisfies $f(U_j) \in [(1-\delta)x, x]$ and that  
$\sum\nolimits_{j \in [\ell]} f(U_j) \geq f(S_2)$. To see why the total function value satisfies this bound, assume for contradiction that the loop terminates with no remaining elements to allocate and that  
$\sum\nolimits_{j \in [\ell]} f(U_j) < f(S_2)$. Since the algorithm partitions $S_2$ into disjoint sets, subadditivity implies $f(S_2) \leq \sum\nolimits_{j \in [\ell]} f(U_j),$ contradicting this assumption. It remains to show that each $U_j$ with $j<\ell$ satisfies $f(U_j) \in [(1-\delta)x, x]$. By sub-additivity and the assumption that  
$f_j(\{i\}) \leq \delta x$ for all $i \in S_2$, we can always add an element to any set with $f(U_j) < (1-\delta)x$ without exceeding $x$. Finally, since each set contributes at least $(1-\delta)x$, the number of sets is at most $ \left\lceil \frac{f(S_2)}{(1-\delta)x}\right\rceil$, as required.

Let $U^{\star}$ be the set that maximizes the demand $U^{\star}\in \argmax_{U_j} ~( f(U_j) -\sum\nolimits_{i\in U_j} p_i)$. Then, we have
\begin{small}
\begin{eqnarray}\label{eq:u-star}
\ell ( f(U^{\star}) -\sum\nolimits_{i\in U^{\star}} p_i  ) \ge \sum\nolimits_{j \in [\ell]} (f(U_j)-\sum\nolimits_{i\in U_j} p_i)\ge f(S_2)  -\sum\nolimits_{i\in \bigcup_{j}U_j}p_i \ge f(S_2) - \sum\nolimits_{i \in S_2 } p_i,
\end{eqnarray}
\end{small}where the first inequality is by optimality of $U^{\star},$ the second inequality is by $\sum\nolimits_{j \in [\ell]}f(U_j)\geq f(S_2)$ and by $\sum\nolimits_{j \in [\ell]}\sum\nolimits_{i\in U_j}p_i=\sum\nolimits_{i\in \bigcup_j U_j}p_i$ since $U_{j_1}\cap U_{j_2}=\emptyset$ $\forall j_1,j_2\in [\ell]$, and the last inequality is by non-negativity of the payments and $\bigcup_j U_j\subseteq S_2.$

We now establish the following inequality:
\begin{eqnarray}\label{eq:before-beta}
f(S^{\star}) -\sum\nolimits_{i \in S^{\star}}p_i 
\le  \left(1+\frac{1}{1-\delta} \right) \Big(f(\bar S) -\sum\nolimits_{i \in \bar S}p_i \Big) +\delta x+\mu (\gamma_2-\gamma_1)f(S^{\star}).
\end{eqnarray}
Note that by Inequality~\eqref{eq:u-star}, we can replace $f(S_2)-\sum\nolimits_{i\in S_2}p_i$ in Inequality~\eqref{eq:s-star} with its upper bound $\ell(f(U^{\star})-\sum\nolimits_{i\in U^{\star}p_i})$ and obtain
$$f(S^{\star}) -\sum\nolimits_{i \in S^{\star}}p_i +\mu (\gamma_1-\gamma_2)f(S^{\star})\le  \mu  (f(S_1) -\sum\nolimits_{i \in S_1}p_i) + \ell(1-\mu) ( f(U^{\star}) -\sum\nolimits_{i \in U^{\star}}p_i) +\delta (1-\gamma_2)x.$$
Let $\bar S$ be the set achieving the maximum between $f(S_1)-\sum\nolimits_{i\in S_1} p_i$ and $f(U^{\star}) -\sum\nolimits_{i \in U^{\star}}p_i$. Since $\gamma_2\ge 0$, we have $f(S^{\star}) -\sum\nolimits_{i \in S^{\star}}p_i +\mu (\gamma_1-\gamma_2)f(S^{\star})\le  (\mu + \ell(1-\mu) )( f(\bar S) -\sum\nolimits_{i \in \bar S}p_i) +\delta x.$ 
To show Inequality~\eqref{eq:before-beta}, it suffices to prove that $\mu + \ell(1-\mu)\leq 1+ \frac{1}{1-\delta}.$ As established above, $\ell\leq \frac{f(S_2)}{(1-\delta)x}+1$. 
This implies that $\mu + \ell(1-\mu)\leq \mu+(\frac{f(S_2)}{(1-\delta)x}+1)(1-\mu)=1+\frac{f(S_2)}{(1-\delta)x}(1-\mu)$. 
Replacing $\mu$ according to Equation~\eqref{eq:mu}, we have $\mu + \ell(1-\mu)\leq 1+\frac{f(S_2)}{(1-\delta)x}(\frac{x-f(S_1)}{f(S_2)-f(S_1)})=1+\frac{1}{(1-\delta)}\frac{f(S_2)(x-f(S_1))}{(f(S_2)-f(S_1))x}.$ By showing that $\frac{f(S_2)(x-f(S_1))}{(f(S_2)-f(S_1))x}\leq 1$ we prove that $\mu + \ell(1-\mu)\leq 1+ \frac{1}{1-\delta}.$ Observe that since $f(S_2)\geq x$, we have that $1-\frac{f(S_1)}{x}\leq 1-\frac{f(S_1)}{f(S_2)}$ which is $\frac{x-f(S_1)}{x}\leq \frac{f(S_2)-f(S_1)}{f(S_2)}$ and hence $\frac{f(S_2)(x-f(S_1))}{(f(S_2)-f(S_1))x}\leq 1.$ This completes the proof of Inequality~\eqref{eq:before-beta}.

Rearranging Inequality~\eqref{eq:before-beta}, and applying $\mu\leq 1, \gamma_2-\gamma_1\leq \delta $ and $f(S^{\star})\leq (1+\delta)x$ we have
$$
\frac{1}{1+\frac{1}{1-\delta}}(f(S^{\star}) -\sum\nolimits_{i \in S^{\star}}p_i )-\frac{(\delta x+\delta (1+\delta)x)}{1+\frac{1}{1-\delta}} \le f(\bar{S})-\sum\nolimits_{i\in \bar S} p_i.
$$
By combining  the last inequality with 
$\frac{(\delta x+\delta (1+\delta)x)}{1+\frac{1}{1-\delta}}=\frac{\delta x (1-\delta)(2+\delta)}{2-\delta}=\frac{\delta x (2-\delta-\delta^2)}{2-\delta}\leq \delta x,$ we obtain 
$$
\frac{1}{1+\frac{1}{1-\delta}}(f(S^{\star}) -\sum\nolimits_{i \in S^{\star}}p_i )-\delta x\le \frac{1}{1+\frac{1}{1-\delta}}(f(S^{\star}) -\sum\nolimits_{i \in S^{\star}}p_i )-\frac{(\delta x+\delta (1+\delta)x)}{1+\frac{1}{1-\delta}} \le f(\bar{S})-\sum\nolimits_{i\in \bar S} p_i,
$$
which concludes the proof of the lemma since $f(S^{\star}) -\sum\nolimits_{i \in S^{\star}}p_i \ge \min\{f(S^{\star}), x\} -\sum\nolimits_{i \in S^{\star}}p_i$.
\end{proof}

\section{Projects Without a Dominant Agent}\label{sub:opt+-approx}
In this section, we prove that an allocation that approximates $\OPT_{+}$—the portion of $\OPT$'s revenue contributed by projects without dominant agent can be computed in polynomial times.

\begin{lemma}[Approximation of $\OPT_{+}$]
\label{lemma:opt-plus} 
For multi-project contract settings with XOS success functions, there exists a polynomial-time algorithm that, using a polynomial number of value and demand queries, computes an allocation $S^{+}$ satisfying $\rev(S^{+}) = \Theta(1) \OPT_{+}$
\end{lemma}

To prove \cref{lemma:opt-plus}, we introduce the following definitions. Let $S^{\star \star}$ represent the allocation derived by restricting the optimal allocation $S^{\star}$ to projects in $I^{\star}$, i.e., 
\begin{equation*}
        S^{\star \star}_j:=\begin{cases}   
       S^{\star}_j & j\in I^{\star},\\
       \emptyset & j\not \in  I^{\star}.
    \end{cases}
\end{equation*}
Note that $\OPT_{+}=\sum\nolimits_{j\in I^{\star}} \rev(S^{\star}_j)=\sum\nolimits_{j\in M} \rev(S^{\star \star}_j)$. 
For any project $j\in M$, let $D_j$ be a discretized set of potential values for $f_j(S^\star)$, and $A_{j,x}$ be the set of agents that contribute at most $\delta x$ to project~$j$ (by themselves). Formally,
\begin{align*}
D_{j} &:=  \{2^k \cdot f_j(\{i\})\mid 0\le k \le \ceil{\log n}, k \in \mathbb{N}, i\in N\}\cup \{0\} &\forall j\in M,\\
A_{j,x}&:=\{ i \mid f_j(\{i\})\leq \delta x \} & \forall j\in M, x\in D_j.
\end{align*}
Our analysis relies on the following linear program.
%
\begin{equation}\label{lp:lp1}\tag{LP1}
\begin{array}{ll@{}ll}
\max & \sum\nolimits_{j\in M} \sum\nolimits_{x\in D_j} \sum\nolimits_{ S\subseteq A_{j,x}} y_{j,x, S} \left(\min\{f_j(S),x\}-\delta\left(1+\frac{1}{1-\delta }\right) x-\sum\nolimits_{i\in S}\frac{\sqrt{c_{ij} x}}{2 \sqrt{2}}\right)\\
\text{s.t.}
&\sum\nolimits_{x\in D_j} \sum\nolimits_{S\subseteq A_{j,x}} y_{j,x, S}\leq 1, \qquad \forall j\in M,\nonumber\\
&\sum\nolimits_{j\in M} \sum\nolimits_{x\in D_j} \sum\nolimits_{  S \subseteq A_{j,x} \mid i\in S}y_{j,x, S}\leq 1,  \qquad \forall i\in N,\nonumber\\
& y_{j,x, S}\in [0,1], \qquad \forall j\in M, x\in D_j, S\subseteq A_{j,x}.\nonumber
\end{array}
\end{equation}

In the above program, each variable $y_{j,x,S}$ acts as the fractional relaxation of an indicator function, 
specifying whether the set of agents~$S$ is allocated to project~$j$ (where $x$ is a parameter capturing an estimate reward of this project). The constraints ensure that the allocation is feasible: The first set of constraints guarantees that each project is assigned to at most one set of agents, while the second set ensures that each agent is assigned to at most one project. The objective function encapsulates two goals: (i)~maximize the overall \emph{capped} reward; 
(ii)~the marginal contribution of every agent with respect to any (fractional) set allocated to project $j$ (i.e., any $y_{j,x,S}>0$ for some $x$) is sufficiently large. Intuitively, (ii)~means that we are in a ``local optimum'', where we don't want to remove any agent from their assigned project. 
The ``prices'' $\sum\nolimits_{i\in S}\frac{\sqrt{c_{ij} x}}{2 \sqrt{2}}$ is similar to the approach done by \cite{DuettingEFK23} with one important necessary modification. In \cite{DuettingEFK23}, they can a-priori guess an estimate of the reward of the (only) project, which allows defining just one set of prices. We cannot estimate the rewards of all projects (as it requires guessing too many values), so we capture all estimates by adding more variables (in particular variables $y_{j,x,S}$ for different levels of guesses $x$) and feasibility constraints to the LP. We mention that the use of a capped function is crucial to the proof. We aim for $f(S)$ to be close to the estimate $x$. As we will see in the rest of the proof, this enables further steps that transform the (approximately) optimal allocation given by LP1 into an allocation that approximates $\text{OPT}^+$.

Solving \ref{lp:lp1} directly as an integer program (if we restrict $y_{j,x,S}$ to be in $\{0,1\}$ instead of in $[0,1]$) is computationally intractable with standard techniques. Instead, we relax the integrality constraints and describe in Section~\ref{sub:frac} how to compute an approximately optimal fractional solution in polynomial time. Then, in Section~\ref{sub:rounding}, we introduce a rounding scheme that converts the fractional solution into an integral one while preserving feasibility and approximately maintaining the same objective value. Then, in Section~\ref{sec:scaling}, we generalize a scaling algorithm developed in \cite{DuettingEFK23} to be apply on each set in the allocation we derived in the former step. This step is useful to obtain a set with approximately the same reward while not paying more than half of the reward.

\subsection{Finding an Approximately Optimal Fractional Solution}\label{sub:frac}

Note that \ref{lp:lp1} consists of a polynomial number of constraints and exponentially many variables. We present an approach for computing an approximately optimal fractional solution to \ref{lp:lp1}, leveraging the dual linear program, the ellipsoid method, and an approximate separation oracle. We establish the construction of this oracle by building on the approximate demand oracle introduced in Claim~\ref{lemma:appx-demand}.

The dual of \ref{lp:lp1} is as follows.
\begin{equation}\label{LP:lp2}\tag{LP2}
\begin{array}{ll@{}ll}
\text{min}&\sum\nolimits_{j \in M} \alpha_j + \sum\nolimits_{i \in N} \beta_i  \\
\text{s.t.}&\alpha_j + \sum\nolimits_{i \in S} \beta_i \geq \min\{f_j(S),x\} -(1+\frac{1}{1-\delta})\delta x- \sum\nolimits_{i \in S} \frac{\sqrt{c_{ij}x}}{2\sqrt{2}} \text{ }\forall j \in M, x \in D_j, S \subseteq A_{j,x},\\
& \alpha_j \geq 0, \quad \forall j \in M,\\
& \beta_i \geq 0, \quad \forall i \in N.
\end{array}
\end{equation}

We first introduce a polynomial-time approximate separation oracle for \eqref{LP:lp2}, which is formally stated in the following claim. 

\begin{claim}[Approximate Separation Oracle]\label{claim:appx-separation}
Given any assignment of the dual variables $ \alpha_j $ for all $ j \in M $ and $ \beta_i $ for all $ i \in N $ corresponding to \eqref{LP:lp2}, there exists a polynomial-time approximate separation oracle that either:  
\begin{enumerate}  
    \item Verifies that all constraints of \eqref{LP:lp2} are satisfied, i.e.,
    $$
    \alpha_j + \sum\nolimits_{i \in S} \beta_i \geq \min\{f_j(S), x\} 
    - \left(1 + \frac{1}{1-\delta}\right) \delta x -  \sum\nolimits_{i \in S} \frac{\sqrt{c_{ij}x}}{2\sqrt{2}}, 
    \quad \forall j \in M, x \in D_j, S \subseteq A_{j,x}.
    $$
    \item Otherwise, it identifies an \emph{approximately} violated constraint. Specifically,
    $$
    \exists j \in M, x \in D_j, S \subseteq A_{j,x} \; \text{such that} \; 
    frac{1}{1 + 1/(1-\delta)} \alpha_j + \sum\nolimits_{i \in S} \beta_i < \min\{f_j(S), x\} 
    -  \sum\nolimits_{i \in S}\frac{ \sqrt{c_{ij}x}}{2\sqrt{2}}.
    $$    
\end{enumerate}  
\end{claim}
\begin{proof}
    We construct the approximate demand oracle using the approximate capped demand derived in \cref{sub:approximate-demad}.
Given an assignment $\{\alpha_j\}_{i\in M},\{\beta_i\}_{i\in N}$ for \eqref{LP:lp2}, for each pair of $ j\in M$ and $x\in D_j$,  we define the prices
\begin{equation*}
    p_i =
    \begin{cases}
        \frac{\sqrt{x}}{2\sqrt{2}} \sqrt{c_{ij}} + \beta_i, & \text{if } f_j(\{i\}) \leq \delta \cdot x, \\
        \infty, & \text{otherwise}
    \end{cases}
\end{equation*}
and compute an approximate demand of function $f_j$ with cap $x$ and prices $p_i$ using \cref{alg:appx-demand}. 
Note that our choice of prices $p_i=\infty$ for $i$ with $f_j(\{i\})>\delta  \cdot x$ ensures that given the prices, the demand or any approximate demand of the agent must be a subset of  $A_{j,x}$.
By \cref{lemma:appx-demand}, 
our approximate demand from \cref{alg:appx-demand} yields a set $\bar{S} \subseteq A_{j,x}$ that satisfies $\min\{f_j(\bar{S}), x\} - \sum\nolimits_{i \in \bar{S}} p_i \geq \frac{1}{1+\frac{1}{1-\delta}} \max_{S} \left(\min\{f_j(S), x_j\} - \sum\nolimits_{i \in S} p_i \right) - \delta \cdot x.$  By rearranging the terms, we obtain:
\begin{equation}\label{eq:appx-demand-rer}
    \left(1 + \frac{1}{1-\delta} \right) \left( \min\{f_j(\bar{S}), x\} - \sum\nolimits_{i \in \bar{S}} p_i \right) 
    \geq \max_{S} \left(\min\{f_j(S), x\} - \sum\nolimits_{i \in S} p_i \right) - \delta \left(1 + \frac{1}{1-\delta}\right) x.
\end{equation}
If $\alpha_j \geq \left(1 + \frac{1}{1-\delta} \right) \left( \min\{f_j(\bar{S}), x\} - \sum\nolimits_{i \in \bar{S}} p_i \right)$ for all $j\in M$ and $x\in D_j$,
all the constraints of \eqref{LP:lp2} are satisfied. This is because by combining the inequality with Inequality~\eqref{eq:appx-demand-rer}, we have 
\begin{eqnarray*}
    \alpha_j  
    \geq  \min\{f_j(S), x\} - \sum\nolimits_{i \in S} p_i - \delta \left(1 + \frac{1}{1-\delta}\right) x,
    &\quad \forall j\in M, x\in D_j, S \subseteq  A_{j,x}.
\end{eqnarray*}

Otherwise, there exists $j\in M$ and $x\in D_j$ such that $\left(1 + \frac{1}{1-\delta} \right) \left( \min\{f_j(\bar{S}), x\} - \sum\nolimits_{i \in \bar{S}} p_i \right) > \alpha_j$. In this case, the constraint for $j$ and $x$ is approximately violated for $\bar{S}$. Specifically, after replacing the prices and rearranging we have
\begin{equation*}
    \min\{f_j(\bar{S}), x\} - \sum\nolimits_{i \in \bar{S}} \frac{\sqrt{x}}{2\sqrt{2}} \sqrt{c_{ij}}  > \frac{1}{1+\frac{1}{1-\delta}} \alpha_j + \sum\nolimits_{i \in \bar{S}} \beta_i.
\end{equation*}
Since there is only a polynomial number of choices for $ j \in M$ and $x \in D_j$, this process terminates in polynomial time.

\end{proof}

\begin{algorithm}[t]
\caption{Approximately Optimal Fractional Solution for \eqref{lp:lp1}\\
\textbf{Input:} Instance of \eqref{lp:lp1}, factor $\epsilon>0$\\
\textbf{Output:} A feasible solution $y$ to \eqref{lp:lp1} satisfying the conditions in \cref{claim:solve-lp1}.
}\label{alg:solve-lp1}
\begin{algorithmic}[1]
\State Derive \eqref{LP:lp2}, the dual of \eqref{lp:lp1}
\State Initialize $\gamma_- \gets 0,\gamma_+ \gets m \cdot n \cdot \max_{i,j}f_j(\{i\})$
\While {$\gamma_+-\gamma_-\ge \epsilon$} \label{state:start-binary}
\State Set $\gamma \gets (\gamma_+ +\gamma_-)/2$ 
\State Add to \eqref{LP:lp2} the constraint $\sum\nolimits_{j \in M} \alpha_j + \sum\nolimits_{i \in N} \beta_i\leq \gamma$  
\State Solve the LP via the ellipsoid method and the approximate separation oracle (Claim~\ref{claim:appx-separation}).
\If{the LP is infeasible}
\State Update the lower bound $\gamma_- \gets \gamma$
\Else
\State Update the upper bound $\gamma_+ \gets \gamma$
\EndIf

\EndWhile\label{slg1:LP-end-binary}
\State Run the ellipsoid method for \eqref{LP:lp2} with additional constraint $\sum\nolimits_{j \in M} \alpha_j + \sum\nolimits_{i \in N} \beta_i<\gamma_-$, introducing a primal variable $y_{j,x,S}$ for each violating constraint $j, x, S$\label{state:run-ellipsoid}
\State Construct a poly-size primal LP using the violated constraints and the corresponding  $y_{j,x,S}$ 
\State Solve this LP obtaining solution $y$\label{state:end-ellipsoid}
\State Scale the solution, $y_{j,x,S}\gets  y_{j,x,S}/(1+1/(1-\delta))$ $\forall j\in M, x\in D_j, x\in A_{j,x}$ 
\If{the objective of \eqref{lp:lp1} at $y$ is negative}\label{state:neg}
\State \Return $y_{j,x,S}\gets 0$ $\forall j\in M, x\in D_j
S\subseteq A_{j,x}$
\EndIf
\While{there exists $y_{j,x,S} > 0$ with $i\in S$ such that $f(i \mid S \setminus i) < \sqrt{x c_{ij}}/2\sqrt{2}$}\label{state:start-improving}
\State Update $y_{j,x,S\setminus \{i\}} \gets y_{j,x,S\setminus \{i\}} + y_{j,x,S}$, and $y_{j,x,S} \gets 0$
\EndWhile
\While{there exists $y_{j,x,S} > 0$ such that $f(S) > (1+\delta)x$}
\State Select an arbitrary agent $i\in S $
\State update $y_{j,x,S\setminus \{i\}} \gets y_{j,x,S\setminus \{i\}} + y_{j,x,S}$, and $y_{j,x,S}\gets 0$
\EndWhile\label{state:end-improving}
\State \Return $y$
\end{algorithmic}
\end{algorithm}

Let $P^{\star} = D^{\star}$ be the optimal solutions to the primal \eqref{lp:lp1} and the dual \eqref{LP:lp2}, respectively. 
\begin{claim}\label{claim:solve-lp1} 
For every $\epsilon>0$, \cref{alg:solve-lp1} computes a feasible solution $y$ to \eqref{lp:lp1} in polynomial time that satisfies
\begin{enumerate}
   \item The number of strictly positive variables $y_{j,x,S}$ is polynomially bounded.\label{item:poly-many}

    \item $\sum\nolimits_{j \in M} \sum\nolimits_{x \in D_j} \sum\nolimits_{S \subseteq A_{j,x}}y_{j,x,S}  \left(\min\{f_j(S),x\} - \frac{\sqrt{x}}{2\sqrt{2}} \sum\nolimits_{i \in S} \sqrt{c_{ij}} \right)\geq \frac{1}{1+1/(1-\delta)} \max\{P^{\star}-\epsilon,0\}.$\label{item:opt-bound}
    
    \item For every $y_{j,x,S}>0$, $f_j(S)\leq (1+\delta) x$, $f_j(\{i\})\leq \delta x$ and $f_j(i\mid S_j\setminus \{i\})\geq \frac{\sqrt{xc_{ij}}}{2\sqrt{2}}$ $\forall i\in S.$\label{item:item-cond} 
\end{enumerate} 
\end{claim}

\begin{proof}
Note that \eqref{lp:lp1} consists of a polynomial number of variables but an exponential number of constraints. Standard duality arguments imply that there exists an optimal solution to \eqref{lp:lp1} with at most $n + m$ nonzero variables. Let $\gamma^{\star}$ denote the minimal objective value for which the dual is feasible. After the binary search in \cref{alg:solve-lp1} ends (line~\ref{slg1:LP-end-binary}), we identify $\gamma_-$ and $\gamma_+$ such that the dual LP is infeasible when constrained to an objective value less than $\gamma_-$ and is feasible when constrained to an objective value less than $\gamma_+$. Furthermore, $\gamma_+ - \gamma_- < \epsilon$. Note that by the approximate separation oracle, any feasible solution corresponds to a feasible solution for \eqref{LP:lp2}, implying that $\gamma_+ \geq D^{\star}$.

The remainder of the algorithm (starting from line~\ref{state:run-ellipsoid}) reconstructs a feasible, approximately optimal solution to \eqref{lp:lp1} using $\gamma_-$. During the ellipsoid method's execution for \eqref{LP:lp2} with the additional constraint that the objective is less than $\gamma_-$, the algorithm identifies polynomially many \emph{approximately} violating constraints (or separating hyperplanes). These constraints correspond to tight constraints in the following program, which is similar to \eqref{LP:lp2} but expressed in the form of the approximate constraints.
\begin{equation*}
\begin{array}{ll@{}ll}
\text{min} & \sum\nolimits_{j \in M} \alpha_j + \sum\nolimits_{i \in N} \beta_i & \\
\text{s.t.} &\frac{1}{1+\frac{1}{1-\delta}} \alpha_j \geq  \min\{f_j(S),x\} -  \sum\nolimits_{i \in S} {(\frac{\sqrt{x}}{2\sqrt{2}}\sqrt{c_{ij}} + \beta_i )}&& \forall j \in M, x \in D_j, S \subseteq A_{j,x},\\
& \alpha_j \geq 0, && \forall j \in M,\\
& \beta_i \geq 0, && \forall i \in N.
\end{array}
\end{equation*}

\noindent
Since there are only polynomially many tight constraints, the dual of this program contains only polynomially many non-slack variables. This results in the following ``small'' problem with polynomially many variables and constraints, which can be solved in polynomial time.
\begin{equation*}\label{LP:dual-dual}\tag{LP3}
\begin{array}{ll@{}ll}
\text{max} &  \sum\nolimits_{j \in M} \sum\nolimits_{x \in D_j} \sum\nolimits_{S \subseteq A_{j,x}} y_{j,x,S}  \left(\min\{f_j(S),x\} - \frac{\sqrt{x}}{2\sqrt{2}} \sum\nolimits_{i \in S} \sqrt{c_{ij}} \right), \\
\text{s.t.} & \sum\nolimits_{x\in D_j}\sum\nolimits_{ S\subseteq A_{j,x}} {\frac{1}{1+1/(1-\delta)}} y_{j,x, S}\leq 1 \quad \forall j\in M,\nonumber\\
&\sum\nolimits_{j\in M}\sum\nolimits_{x\in D_j}\sum\nolimits_{  S \subseteq A_{j,x} \mid i\in S}y_{j,x, S}\leq 1  \quad \forall i\in N,\nonumber
\end{array}
\end{equation*}
By duality, since the dual is infeasible for $\gamma_-$, the reduced primal LP must have a feasible solution with an objective value of at least $\gamma_-$, which can be computed in polynomial time. To obtain a feasible solution for the original problem \eqref{lp:lp1}, we apply a scaling factor of $\frac{1}{1+1/(1-\delta)}$ to the solution obtained by solving the reduced primal LP. This ensures a feasible solution $y$ to \eqref{lp:lp1} such that
\begin{eqnarray*}
&\sum\nolimits_{j \in M} \sum\nolimits_{x \in D_j} \sum\nolimits_{S \subseteq A_{j,x}} y_{j,x,S}  \left(\min\{f_j(S), x\} - \frac{\sqrt{x}}{2\sqrt{2}} \sum\nolimits_{i \in S} \sqrt{c_{ij}} \right)\\ 
&\geq\frac{1}{1+\frac{1}{1-\delta}}\gamma_- \geq \frac{1}{1+\frac{1}{1-\delta}}(\gamma_+-\epsilon)\geq \frac{1}{1+\frac{1}{1-\delta}}(P^{\star}-\epsilon)
\end{eqnarray*}
As we now show, the final steps of the algorithm (starting for line~\ref{state:start-improving}) further improve the objective solution while preserving feasibility and at most the same number of variables. Therefore, the proof for \eqref{item:poly-many} and \eqref{item:opt-bound} is completed. To conclude, we verify that \eqref{item:item-cond} still holds after the final two loops terminate.

Since $y$ is feasible for \eqref{LP:dual-dual} and, by definition of $A_{j,x}$, we have $f_j(\{i\}) \leq \delta x$ for all $i \in S$ whenever $y_{j,x,S} > 0$. Moreover, the number of iterations in the while loops is bounded by $M \cdot N \cdot$ the number of strictly positive $y_{j,x,S}$, which is polynomially bounded.

Removing agents from $S$ when $f(S) > (1+\delta)x$ improves the objective: By subadditivity, and the fact that $f_j({i}) \leq \delta x$, removing an agent when $f(S) > (1+\delta)x$ ensures that $f(S) \geq x$. This adjustment reduces the payments $\sum\nolimits_{i\in S} \frac{\sqrt{x c_{ij}}}{2\sqrt{2}}$ without altering $\min{f(S),x}$. Consequently, we obtain $f(S) \leq (1+\delta)x$.

Finally, removing $i \in S$ such that 
$f_j(i \mid S \setminus \{i\}) < {\sqrt{x c_{ij}}}/{2\sqrt{2}}$ increases the objective, as the agent's contribution is smaller than the corresponding payment. Note that the final solution remains feasible and satisfies all required conditions, completing the proof.
\end{proof}

\begin{claim}\label{claim:opt-sol-to-opt}
The solution computed by Algorithm~\ref{alg:solve-lp1} for any $\epsilon < \delta \min_{j,i} f_j(\{i\})$, satisfies
$$
\frac{1}{5} \sum\nolimits_{j \in M} f_j(S^{\star \star}_j) \leq \sum\nolimits_{j \in M} \sum\nolimits_{x \in D_j} \sum\nolimits_{S \subseteq A_{j,x}} y_{j,x,S} f_j(S).
$$
\end{claim}

The key idea is to show that $S^{\star \star}$ forms a feasible solution to \eqref{lp:lp1}. We then bound its objective in \eqref{lp:lp1} from above using the approximately optimal solution computed by \cref{alg:solve-lp1} (see \cref{claim:solve-lp1}), which corresponds to the right-hand side of the claim statement.  

Furthermore, by applying Lemma 3.3 from \cite{duetting2024multiagent} (see \cref{lemma:lemma33-multi})—which establishes a connection between the prices  
${\sqrt{x c_{ij}}}/{2\sqrt{2}} $ in \eqref{lp:lp1} and $ f_j(S^{\star\star}_j) $—we derive a lower bound on the objective of \eqref{lp:lp1} in terms of  
$ \sum\nolimits_{j\in M} f_j(S^{\star\star}_j) $, which corresponds to the left-hand side.

\begin{proof}[Proof of Claim~\ref{claim:opt-sol-to-opt}]
    First,
\begin{eqnarray*}
    \sum_{j \in M} \sum_{x \in D_j} \sum_{S \subseteq A_{j,x}}y_{j,x,S}  {f_j(S)}\geq\sum_{j \in M} \sum_{x \in D_j} \sum_{S \subseteq A_{j,x}}y_{j,x,S}  \left(\min\{f_j(S),x\} -  \sum_{i \in S}\frac{\sqrt{c_{ij} x}}{2\sqrt{2}}  \right),
\end{eqnarray*}
since $c_{ij}\geq 0$ and $x\geq 0$ for all $i\in N, j\in M$ and $x\in D_j$.
By \cref{claim:solve-lp1}, we have 
$$\sum_{j \in M} \sum_{x \in D_j} \sum_{S \subseteq A_{j,x}}y_{j,x,S}  {f_j(S)}\geq \frac{1}{1+1/(1-\delta)}(P^{\star}-\epsilon).$$
Next, we show that $S^{\star \star}$ induces a feasible solution to \eqref{lp:lp1}. For every project $j\in I^{\star}$, let $i_j=\arg\max_{i\in S^{\star \star}_j} f_j(\{i\})$. By monotonicity and subadditivity, $f_j(\{i_j\})\le f_j(S^{\star \star}_j) \le n f_j(\{i_j\})$. Therefore, for every project $j\in I^{\star}$ there exists $x'_j\in D_j$ such that 
\begin{eqnarray}\label{eq:x-up-low}
\frac{x'_j}{2}\leq f_j(S^{\star \star}_j)\leq x'_j.
\end{eqnarray}
Furthermore, by definition of $S^{\star \star}_j$, we have that $f_j(\{i\})\leq \delta f_j(S^{\star \star}_j) \leq \delta x'_j$ $\forall i\in S^{\star \star}_j$, that is, $S^{\star \star}_j \subseteq A_{j,x'_j}.$ Taking $y_{j,x,S}=1$ if $x=x'_j,$ $S=S^{\star \star}_j$, and $y_{j,x,S}=0$ otherwise gives a feasible solution to $\eqref{lp:lp1}$ with objective bounded by $P^{\star}$ as follows.
\begin{align*}
P^{\star}&\geq \sum_{j \in M}\bigg(\min\{f_j(S^{\star\star}_j),x'_j\}-\delta(1+\frac{1}{1-\delta})x'_j-\sum_{i\in S^{\star \star}_j}\frac{\sqrt{c_{ij}x'_j}}{2 \sqrt{2}}\bigg)\\
&\geq \sum_{j \in M}\bigg(f_j(S^{\star\star}_j)-\delta(1+\frac{1}{1-\delta})2f_j(S^{\star\star}_j)-\frac{1}{2}\sum_{i\in S^{\star \star}_j}\sqrt{c_{ij}f_j(S^{\star\star}_j)}\bigg)\\
&\geq\sum_{j \in M}\bigg(\frac{1}{2}f_j(S^{\star\star}_j)-\delta(1+\frac{1}{1-\delta})2f_j(S^{\star\star}_j)\bigg)=(\frac{1}{2}-2\delta (\frac{2-\delta}{1-\delta}))\sum_{j\in M}f_j(S^{\star\star}_j).
\end{align*}
where the second inequality follows from \eqref{eq:x-up-low} and the third inequality holds by Lemma 3.3 in \cite{DuettingEFK23} (see Lemma~\ref{lemma:lemma33-multi}). 
Note that if $\sum_{j\in M} f_j(S^{\star \star}_j) = 0$, then, since we always return a nonnegative solution (line~\ref{state:neg}), the condition of the claim holds trivially. If $\sum_{j\in M} f_j(S^{\star \star}_j) > 0$, then by monotonicity of $f_j$ it must also hold that $\sum_{j\in M}f_j(S^{\star \star}_j)>\min_{j,i}f_j(\{i\})$. Therefore, Since $\epsilon<\delta \min_{j,i}f_j(\{i\}),$ 
$P^{\star}-\epsilon \geq (\frac{1}{2}-2\delta (\frac{2-\delta}{1-\delta}))\sum_{j \in M}f_j(S^{\star\star}_j)-\delta \min_{i,j}f_j(\{i\}) \geq (\frac{1}{2}-2\delta (\frac{2-\delta}{1-\delta})-\delta)\sum_{j \in M}f_j(S^{\star\star}_j)$.
Then, by Claim~\ref{claim:solve-lp1} we have that 
\begin{align*}
    \sum_{j \in M} \sum_{x \in D_j} \sum_{S \subseteq A_{j,x}}y_{j,x,S} f_j(S) 
    \geq \sum_{j \in M} \sum_{x \in D_j} \sum_{S \subseteq A_{j,x}}y_{j,x,S}  \left(\min\{f_j(S),x\} - \frac{\sqrt{x}}{2\sqrt{2}} \sum_{i \in S} \sqrt{c_{ij}} \right)\\
    \geq\frac{1}{1+1/(1-\delta)} (P^{\star}-\epsilon)\geq \frac{1}{1+1/(1-\delta)}(\frac{1}{2}-2\delta (\frac{2-\delta}{1-\delta})-\delta) \sum_{j \in M}f_j(S^{\star \star}).
\end{align*}
Since $\delta =\frac{1}{129}$, we have than $\frac{1}{1+1/(1-\delta)}\geq 0.49$, and $\frac{1}{2}-2\delta (\frac{2-\delta}{1-\delta})-\delta \geq 0.46$. This implies that $\sum_{j \in M} \sum_{x \in D_j} \sum_{S \subseteq A_{j,x}}y_{j,x,S} f_j(S)\geq 0.22\cdot\sum_{j \in M}f_j(S^{\star \star})$
This completes the proof.
\end{proof}

\subsection{Rounding the Fractional Solution}\label{sub:rounding}

In this section, we present a simple and intuitive algorithm for rounding a fractional solution to \eqref{lp:lp1} into an integral solution while preserving key properties that will be crucial in later analysis. We begin by noting that any fractional solution to \eqref{lp:lp1} can be interpreted as a distribution over sets. Based on this observation, we introduce a rounding technique that converts these random sets into deterministic sets while ensuring that at least half of the reward from the fractional solution is preserved. Furthermore, each set in the integral solution remains a subset of some original set in the distribution.\footnote{There are also randomized approaches for rounding fractional solutions (e.g., \cite{welfare-maximizing-uri-feige}). While we are aware of these methods, for completeness, we present our deterministic rounding method for XOS functions.}


\begin{definition}\label{obs:frac-to-dist}
Given a feasible solution $y$ to \eqref{lp:lp1}. Define a set of distributions $F_1,\ldots,F_m$ as follows.
\begin{eqnarray*}
{\Pr}_{\tilde{S}_j\sim F_j}[\tilde{S}_j=S]=\begin{cases}
\sum\nolimits_{x\in D_j} y_{j,x,S} & S\neq \emptyset\\
1-\sum\nolimits_{S'\subseteq N \mid S'\neq\emptyset}\sum\nolimits_{x\in D_j} y_{j,x,S'} & S=\emptyset
\end{cases}
\end{eqnarray*}
\end{definition}
This is a valid distribution by the first set of constraints in \eqref{lp:lp1}. Note that the defined distribution satisfied the following properties. 
\begin{enumerate}
    \item For every agent $i \in N$,  the total probability of being selected across all projects is bounded by $1$, i.e., $\sum\nolimits_{j\in [m]} \Pr_{\tilde{S}_j \sim F_j}[i \in \tilde{S}_j] \leq 1.$ This is by the second set of constraints in \eqref{lp:lp1}
    \item For every project $j \in M$, the expected function value under $F_j$ is equal to the weighted sum over the fractional solution $\sum\nolimits_{x\in D_j}\sum\nolimits_{S\subseteq A_{j,x}} y_{j,x,S} f_j(S) = \mathbb{E}_{\tilde{S} \sim F_j}[f_j(\tilde{S})].$ This is by construction and since $f(\emptyset)=0$
    \item Every set $S\neq \emptyset$ in the support of $F_j$ corresponds to some $x \in D_j$ with $y_{j,x,S} > 0$. This is also by construction.
\end{enumerate}

We now present \cref{alg:rounding}, which receives XOS functions $f_1,\ldots,f_m$ over $N$, and distributions $F_1,\ldots,F_m$ over $N$, and returns a (partial) allocation.
Each distribution $F_j$ is supported by $k_j$ sets $\tilde{S}_{j,1}, \ldots,\tilde{S}_{j,k_j}$, where we denote $ q_{j,t} = \Pr_{S\sim F_j}[S=\tilde{S}_{j,t}]$ for every $\tilde{S}_{j,t}$. 
Our algorithm uses the following notation:
We define the marginal distribution $F_j^{-S}$ to be such that with probability $q_{j,t}$ it is $\tilde{S}_{j,t} \setminus S$. This represents a modified distribution in which we remove  $S$ from each set in the support. 
We also use the following definition given subsets $M'\subseteq M$, and $N'\subseteq N$, let 
\begin{equation}
    \VAL(M',N')  = \sum\nolimits_{j\in M'} \sum\nolimits_{t=1}^{k_j} q_{j,t} \cdot f_{j}(\tilde{S}_{j,t} \cap N') . \label{eq:val}
\end{equation} 
Note that $\VAL(\cdot,\cdot)$ can be computed in polynomial time in the sum of the sizes of the supports of $F_1,\ldots,F_m$.
The algorithm's key idea is to find $j\in M$ which can be allocated one set in its corresponding support such that we can continue recursively. The main challenge is to prove that such a $j$, together with a set in its support always exist. The process then continues by solving the smaller problem without the project and the allocated agents.
This is formalized in the following. 

\begin{algorithm}[t]
\caption{Deterministic Rounding for XOS Functions\\ \textbf{Input:} $N,M$, XOS functions $f_1,\ldots, f_m$ over $N$, distributions $F_1,\ldots,F_m$ over subsets of $N$\\
\textbf{Output:} A partial allocation $T_1\ldots,T_m$}\label{alg:rounding}
\begin{algorithmic}[1]
\If{$M=\emptyset$ \label{state:existsj}}
\State \Return the empty allocation
\EndIf
\State $\tilde{S}_{j,1}, \ldots,\tilde{S}_{j,k_j} \gets \text{Supp}(F_j)$ for every $j\in M$\label{state:tilde-s}
\State $ q_{j,t} \gets \Pr_{S\sim F_j}[S=\tilde{S}_{j,t}]$ for every $\tilde{S}_{j,t}$\label{state:q-jt}
\For{$j \in M$} 
\For{$t = 1\ldots k_j$ }
\State Let $N'=N\setminus \tilde{S}_{j,t}$, $M'=M\setminus \{j\}$ \label{state:remove-agent-project}
\If{$f_j(\tilde{S}_{j,t})  +   \VAL(M',N')/2 \geq  \VAL(M,N)/2 $\label{state:if3}}  
\State Let $T_j = \tilde{S}_{j,t}$
\State Let $F_{j'}'=F_{j'}^{-\tilde{S}_{j,t}}$ for $j' \in M' $
\State Apply Alg.~\ref{alg:rounding} recursively on $N'$, $M'$,$\{f_{j'}\}_{j'\in M'}$, $\{F_{j'}'\}_{j'\in M'}$ to get allocation $\{T_{j'}\}_{j'\in M'}$
\State \Return $T_1,\ldots,T_m$
\EndIf
\EndFor
\EndFor
\end{algorithmic}
\end{algorithm}


\begin{lemma}[Deterministic Rounding of XOS Functions]\label{lemma:rounding}
Algorithm~\ref{alg:rounding} finds a partial partition $T_1,\ldots,T_m$ in polynomial time (in $n,m$ and in the size of the union of the supports of $F_1,\ldots,F_m$) such that for all $j \in M$, there exists $S_j$ in the support of $F_j$ such that $T_j \subseteq S_j$.
Moreover, it holds that 
$$\sum\nolimits_{j\in M} f_j(T_j) 
\geq \VAL(M,N)/2.$$
\end{lemma}

\begin{proof}
    Let $a_{j,t}:N \rightarrow \R_{\geq 0 }$ be the supporting\footnote{Such a function exists by the XOS definition. Algorithm~\ref{alg:rounding} does not need to find this supporting function, it is only used for the correctness proof.} additive function for the XOS function $ f_j$ with respect to set $\tilde{S}_{j,t}$, i.e., $ \sum\nolimits_{i\in \tilde{S}_{j,t}} a_{j,t}(i) = f_j(\tilde{S}_{j,t}) $ and for every $S\subseteq N$, $f_{j}(S) \geq \sum\nolimits_{i\in S} a_{j,t}(i)$.
    We have that 
    \begin{equation}    
\VAL(M,N) = \sum\nolimits_{j\in M}\sum\nolimits_{t=1}^{k_j} q_{j,t}  \cdot f_j(\tilde{S}_{j,t})  = \sum\nolimits_{j=1}^m \sum\nolimits_{t=1}^{k_j} q_{j,t} \sum\nolimits_{i \in \tilde{S}_{j,t}} a_{j,t}(i).  \label{eq:contra-opt}  
     \end{equation}
We define a price $p_i$ for each $i \in N$ and $p_j$ for each $j \in M$ as follow: 
\begin{eqnarray}
p_i = \sum\nolimits_{j=1}^m \sum\nolimits_{t=1}^{k_j} q_{j,t} \cdot \indicator{i\in \tilde{S}_{j,t}} \cdot a_{j,t}(i) & \forall i\in N,\label{eq:pi}\\
\ell_j = \sum\nolimits_{t=1}^{k_j} q_{j,t} \sum\nolimits_{i \in \tilde{S}_{j,t}} a_{j,t}(i) & \forall j\in M.\label{eq:ellj}
\end{eqnarray}
One can think of $p_i$ (resp. $\ell_j$) as the terms corresponding to $i$ (resp. $j$) in the RHS of Equation~\eqref{eq:contra-opt}.  
Note that, \begin{equation}
\sum\nolimits_{i\in N} p_i = \sum\nolimits_{j\in M} \ell_j = \VAL(M,N) \label{eq:contra-half}
\end{equation}

To complete the proof we show that the algorithm is well defined (which means that always for some iteration defined by $j\in M$ and $t\in\{1,\ldots,k_j\}$, the inequality of Line~\ref{state:if3} holds).
The guarantees of the lemma are immediate from the condition in the inequality of Line~\ref{state:if3}, and since by definition, $T_j$ it is a subset of a set in the support.

Assume towards contradiction that for every $j$, and every $t$ it holds that 
\begin{equation}
f_j(\tilde{S}_{j,t}) + \VAL(M\setminus\{j\},N \setminus \tilde{S}_{j,t})/2  < \VAL(M,N)/2 . \label{eq:contra} \nonumber
\end{equation}
By rearrangement we get that 
\begin{eqnarray}
     &&2\cdot f_j(\tilde{S}_{j,t})  < \VAL(M,N) - \VAL(M\setminus\{j\},N \setminus \tilde{S}_{j,t})\nonumber \\ & = & \sum\nolimits_{j'\in M}\sum\nolimits_{t'=1}^{k_{j'}} q_{j',t'} \cdot f_{j'}(\tilde{S}_{j',t'}) - \sum\nolimits_{j' \in M \setminus \{j\}}\sum\nolimits_{t'=1}^{k_{j'}} q_{j',t'} \cdot f_{j'}(\tilde{S}_{j',t'} \setminus \tilde{S}_{j,t} ) \nonumber\\ & \leq &  \sum\nolimits_{j'\in M}\sum\nolimits_{t'=1}^{k_{j'}} q_{j',t'} \sum\nolimits_{i\in \tilde{S}_{j',t'}} a_{j',t'}(i) - \sum\nolimits_{j' \in M \setminus \{j\}}\sum\nolimits_{t'=1}^{k_{j'}} q_{j',t'} \sum\nolimits_{i\in \tilde{S}_{j',t'} \setminus \tilde{S}_{j,t}  }a_{j',t'}(i)  
     \nonumber \\ & =  & \ell_j + \sum\nolimits_{j'\in M\setminus \{j\}} \sum\nolimits_{t'=1}^{k_{j'}} q_{j',t'} \sum\nolimits_{i\in \tilde{S}_{j',t'} \cap \tilde{S}_{j,t}} a_{j',t'}(i) \nonumber  \\ &  \leq &   \ell_j + \sum\nolimits_{i\in  \tilde{S}_{j,t}} \sum\nolimits_{j'\in M} \sum\nolimits_{t'=1}^{k_{j'}} q_{j',t'} \cdot   \indicator{i\in \tilde{S}_{j',t'}} \cdot a_{j',t'}(i) 
     =     \ell_j + \sum\nolimits_{i \in \tilde{S}_{j,t}} p_i , \label{eq:contra3} 
     \end{eqnarray}
     where the second inequality is by the XOS property, and the last inequality is since we added the term corresponding to $j'=j$.
We sum Inequality~\eqref{eq:contra3} over all $j$ and $t$ with coefficient $q_{j,t}$, and overall, we get that 
\begin{eqnarray}  
2\sum_{j\in M} \sum_{t=1}^{k_j} q_{j,t  }\cdot f_j(\tilde{S}_{j,t})   <  \sum_{j\in M} \sum_{t=1}^{k_j} q_{j,t  }\cdot \left(\ell_j + \sum_{i\in \tilde{S}_{j,t}} p_i \right) \leq  \sum_{j\in M} \ell_j + \sum_{j\in M} \sum_{t=1}^{k_j} q_{j,t  }\cdot \left( \sum_{i\in \tilde{S}_{j,t}} p_i \right)  \nonumber 
\\  =   \sum_{j\in M} \ell_j  + \sum_{i\in N} p_i \cdot \sum_{j\in M} \sum_{t=1}^{k_j} q_{j,t  }\cdot  \indicator{i\in \tilde{S}_{j,t}}  \nonumber  \leq   \sum_{j\in M} \ell_j  + \sum_{i\in N} p_i   \stackrel{\eqref{eq:contra-half}}{=}  2 \cdot \VAL(M,N), \nonumber 
\label{eq:contra2}
\end{eqnarray}
where the first inequality follows from summing the terms as described above, the second inequality holds because $\sum_{t=1}^{k_j} q_{j,t}=  1$, the first equality results from rearranging the terms, and the final inequality follows from Equation~\eqref{eq:contra-half}. This is a contradiction since $2\sum_{j\in M} \sum_{t=1}^{k_j} q_{j,t  }\cdot f_j(\tilde{S}_{j,t})  =2\cdot \VAL(M,N)$. Thus, the algorithm is well-defined, which concludes the proof.

\end{proof}

\subsection{Scaling the Rounded Solution}\label{sec:scaling}

In the remainder of this section, we apply a generalized scaling lemma (see Lemma~\ref{lemma:scalinglemma-new}) to the rounded solution, yielding an allocation that guarantees a constant-factor approximation to $ \OPT^{+} $.\footnote{We detail how the generalized scaling lemma differs from the one introduced by \citet{DuettingEFK23} in \cref{appx:scaling-lemma}.
} Specifically, after applying the rounding procedure in \cref{alg:rounding} to a fractional solution of \eqref{lp:lp1} that is computed by Algorithm~\ref{alg:solve-lp1} with $ \epsilon < \delta \min_{i,j} f_j(\{i\}) $, we obtain a (partial) allocation $ T_1, \dots, T_m $. By \cref{lemma:rounding} and \cref{claim:opt-sol-to-opt}, this allocation satisfies.
$$
\sum\nolimits_{j\in M} f_j(T_j) \geq \frac{1}{10}\sum\nolimits_{j\in M} f_j(S^{\star\star}_j).
$$
Furthermore, by \cref{claim:solve-lp1}, for all $ j \in M $, there exist sets $ S_j $ and values $ x_j $ such that $ T_j \subseteq S_j $, $ f_j(S_j) \leq (1+\delta) x_j $, $ f_j(\{i\}) \leq \delta x_j $, and $ f_j(i \mid S_j \setminus \{i\}) \geq {\sqrt{x_j c_{ij}}}/{2\sqrt{2}} $ for all $ i \in S_j $. We show that under these conditions an approximately optimal allocation can be computed in polynomial time. 

\begin{claim}\label{claim:apply-scaling}
Given an allocation of agents to projects $T_1, \dots, T_m$ and a constant $\kappa>0$, satisfying $$\sum\nolimits_{j\in M}f_j(T_j)\geq \kappa \sum\nolimits_{j\in M}f_j(S^{\star\star}_j),$$ and assuming that for every project $j\in M$ there exists $S_j$ and $x_j$ such that $T_j\subseteq S_j,$ $f_j(S_j)\leq (1+\delta) x_j$, and $f_j(i\mid S_j\setminus \{i\})\geq {\sqrt{x_jc_{ij}}}/{2\sqrt{2}}$ $\forall i\in S_j$.
An allocation $U = (U_1, \dots, U_m)$ satisfying $$\rev(U) \geq \frac{\kappa}{512} \rev(S^{\star \star}),$$
can be computed in polynomial time.
\end{claim}

\begin{proof}
    We apply the generalized scaling lemma to each project individually (see Lemma~\ref{lemma:scalinglemma-new}). This lemma considers a set $T \subseteq S$ and shows that its value, $f(T)$, can be reduced to any desired level $\Psi$ by removing certain elements while ensuring that the remaining elements retain sufficiently high marginal values with respect to their original values in the superset $S$ of $T$. 
Setting $\Psi = \frac{1}{128} f_j(T_j)$ in \cref{lemma:scalinglemma-new} for each project $j\in M$, we obtain a subset $U_j \subseteq T_j$ satisfying
\begin{eqnarray}\label{eq:upp-bound-scaling}
\frac{1}{256}f_j(T_j)\leq f_j(U_j)\leq \frac{1}{128}f_j(T_j)+\max_{i\in T_j}f_j(\{i\}),
\end{eqnarray}
and 
\begin{eqnarray}\label{eq:bound-marginal}
f(i\mid U_j\setminus \{i\})\geq \frac{1}{2}f(i\mid S_j\setminus \{i\}) \quad \text{ for all } i \in U_j.    
\end{eqnarray}
Using the claim's assumptions on the existence of $S_j$ and $x_j$, we further bound Inequality~\eqref{eq:upp-bound-scaling} as \begin{eqnarray}\label{eq:u-bound}
f_j(U_j)\leq \frac{1}{128}f_j(T_j)+\max_{i \in T_j}f_j(\{i\})\leq \frac{1}{128}(1+\delta)x_j+\delta x_j,
\end{eqnarray}
where the second inequality follows from $T_j \subseteq S_j$, the assumptions $f_j(S_j) \leq (1+\delta) x_j$ and $f_j(\{i\}) \leq \delta x_j$ for all $i \in N$. Similarly, for every $i \in U_j$, by Inequality~\eqref{eq:bound-marginal} we have the following bound.
\begin{eqnarray*}
f(i\mid U_j \setminus \{i\})   \stackrel{\eqref{eq:bound-marginal}}{\geq} \frac{1}{2}f(i\mid S_j\setminus \{i\})
\geq \frac{\sqrt{x_j \cdot c_{ij}}}{4\sqrt{2}} \stackrel{\eqref{eq:u-bound}}{\geq} \frac{\sqrt{f_j(U_j) \cdot c_{ij}}}{4\sqrt{2(\frac{1}{128}(1+\delta )+\delta )}} \geq \sqrt{2 c_{ij} \cdot f_j(U_j)}
\end{eqnarray*}
where the second inequality follows from the assumption that $f_j(i \mid S_j \setminus \{i\}) \geq \sqrt{x_j c_{ij}}/2\sqrt{2}$ for all $i \in S_j$, and the last  inequality is since $\delta \leq \frac{1}{129}$. 

Thus, $U_j$ satisfies the requirements of  Lemma~\ref{lemma:lemma34-multi}. Therefore, by Lemma~\ref{lemma:lemma34-multi} it holds that $$\rev_j(U_j)\geq \frac{f_j(U_j)}{2} \stackrel{\eqref{eq:upp-bound-scaling}}{\geq} \frac{f_j(T_j)}{512}.$$ 
Summing over all projects, we get that  
\begin{eqnarray*}
\sum_{j \in M}\rev_j(U_j)& \geq & \sum_{j \in M}\frac{1}{512}f_j(T_j) \geq \sum_{j \in M}\frac{\kappa}{512} f_j(S^{\star \star}_j) 
 \geq  \sum_{j \in M}\frac{\kappa }{512}\rev_j(S^{\star \star}_j)  = \frac{\kappa}{512} \rev(S^{\star \star}),
\end{eqnarray*}
where the second inequality follows from the claim's assumption, the third inequality holds because agent payments are always nonnegative, implying that revenue is upper-bounded by the reward. This completes the proof. 
\end{proof}

\section{Putting it All Together}
We are now ready to conclude our arguments by consolidating the results from the previous sections into a single polynomial-time algorithm that computes an approximately optimal allocation as follows.


\begin{algorithm}[H]
\caption{Constant Factor Approximation for $\OPT$.\\
\textbf{Input:} Set of agents $N$, project success functions $f_1, \dots, f_m$, and agent costs $c_{ij}$ for every $i \in N,j \in M$\\
\textbf{Output:} Allocation $S^-$ that approximates $\OPT$
}\label{alg:appx-opt}
\begin{algorithmic}[1]
\State Run \cref{alg:appx-opt-minus} to compute an allocation $S^-$ that approximates $\OPT^-$ by a constant factor\label{state:apply-appx-opt-minus}
\State Formulate the linear program~\eqref{lp:lp1}
\State Run \cref{alg:solve-lp1} with $\epsilon<\delta \min_{j,i}f_j(\{i\})$ to compute an approximately optimal fractional solution $y$ to \eqref{lp:lp1}\label{state:apply-solve-lp1}
\State Construct distributions $\{F_j\}_{j\in M}$ over subsets of $N$ from $y$ as described in \cref{obs:frac-to-dist} \label{state:dist}
\State Run \cref{alg:rounding} to convert the random allocation defined by $\{F_j\}_{j\in M}$ into a deterministic allocation $T$ \label{state:derandomize}
\State Apply \cref{claim:apply-scaling} to $T$ to obtain an allocation $S^+$ that approximates $\OPT^+$ by a constant factor \label{state:do-scaling}
\State \Return $\arg\max_{S\in \{S^-,S^+\}}\rev(S)$ \label{state:return-best}
\end{algorithmic}
\end{algorithm}

\begin{proof}[Proof of Theorem~\ref{thm:const-diff-projects}]
By \cref{lem:opt-minus}, running Alg.~\ref{alg:appx-opt-minus} in line~\ref{state:apply-appx-opt-minus} computes an allocation $S^-$ s.t. $$\rev(S^-)\geq \delta \cdot \OPT^-$$
As established in \cref{claim:solve-lp1} and \cref{claim:opt-sol-to-opt}, running \cref{alg:solve-lp1} in line~\ref{state:apply-solve-lp1} with $\epsilon<\delta \min_{j,i}f_j(\{i\})$ computes in polynomial time a solution $y$ to \eqref{lp:lp1} with the following properties.
\begin{enumerate}
    \item[(1.1)] The number of positive variables $y_{j,x,S}$ is polynomially bounded. 
    \item[(1.2)] For every $y_{j,x,S}>0$, $f_j(S)\leq (1+\delta) x$, $f_j(\{i\})\leq \delta x$ and $f_j(i\mid S_j\setminus \{i\})\geq \frac{\sqrt{xc_{ij}}}{2\sqrt{2}}$ $\forall i\in S.$
    \item[(1.3)] The weighted reward according to $y$ is at least a constant fraction of the reward of $S^{\star \star}$. I.e., 
    \begin{eqnarray}
    \frac{1}{5}\sum_{j \in M}f_j(S^{\star \star}_j)\leq \sum_{j \in M} \sum_{x \in D_j} \sum_{S \subseteq A_{j,x}}y_{j,x,S}  {f_j(S)}. \label{eq:final-lp-sol}    
    \end{eqnarray}
\end{enumerate}
Since the number of positive variables $y_{j,x,S}$ is polynomially bounded, we can construct in polynomial time, in line~\ref{state:dist}, the distributions $\{F_j\}_{j\in M}$ as described in \cref{obs:frac-to-dist}. These distributions satisfy the following properties.
\begin{enumerate}
\item[(2.1)] For every agent $i \in N$,  the total probability of being selected across all projects is bounded by $1$, i.e., $\sum_{j\in [m]} \Pr_{\tilde{S}_j \sim F_j}[i \in \tilde{S}_j] \leq 1.$
\item[(2.2)] For every project $j \in M$, the expected function value under $F_j$ is equal to the weighted sum over the fractional solution 
\begin{eqnarray}\label{eq:final-dist}
\sum_{x\in D_j}\sum_{S\subseteq A_{j,x}} y_{j,x,S} f_j(S) = \mathbb{E}_{\tilde{S} \sim F_j}[f_j(\tilde{S})].
\end{eqnarray} 
\item[(2.3)] Every set $S\neq \emptyset$ in the support of $F_j$ corresponds to some $x \in D_j$ with $y_{j,x,S} > 0$.
\label{item:every-set}
\end{enumerate}
Next, in line~\ref{state:derandomize}, we apply \cref{alg:rounding} on $\{F_j\}_{j\in M}$, by \cref{lemma:rounding}, it runs in polynomial time in $M,N$ and the size of the (polynomial) support. Since the distributions satisfy the conditions of \cref{lemma:rounding}, we obtain an allocation $T=(T_1,\ldots,T_m)$ such that for all $j \in M$, there exists $S_j$ in the support of $F_j$ such that $T_j \subseteq S_j$. This implies, by properties (2.3) and (1.2), that for every project $j\in M$, there exists $x_j$ such that $f_j(S)\leq (1+\delta) x_j$, $f_j(\{i\})\leq \delta x_j$ and $f_j(i\mid S_j\setminus \{i\})\geq \frac{\sqrt{x_jc_{ij}}}{2\sqrt{2}}$ $\forall i\in S$. Furthermore, we have $\sum_{j\in M} f_j(T_j) 
\geq \frac{1}{2}\sum_{j\in M}\E_{S_j \sim F_j}[f_j(S_j)].$   
Combining this with Inequality~\eqref{eq:final-lp-sol} and Equality~\eqref{eq:final-dist}, we have an allocation $T$ that satisfy
\begin{eqnarray}
\sum_{j\in M} f_j(T_j) 
\geq \frac{1}{10}f_j(S^{\star \star}_j).    
\end{eqnarray} 
Let $\kappa=\frac{1}{10}$. We can apply \cref{claim:apply-scaling} in line~\ref{state:do-scaling} since all the conditions of the claim are satisfied. As a result, we obtain an allocation $S^+$ such that $\rev(S^+)\geq \frac{1}{512}\frac{1}{10}\rev(S^{\star \star})=OPT^+.$
By picking the best of these two solutions in line~\ref{state:return-best}, the expected revenue of the returned allocation satisfies
$$
\max_{S\in \{S^+,S^-\}}\rev(S) \geq \frac{\rev(S^+)+\rev(S^-)}{2}=\Theta(1) \cdot (\OPT^+ +\OPT^-)=\Theta(1) \cdot \OPT.
$$
This completes the proof.
\end{proof}



\section*{Acknowledgments}
T.~Alon and I.~Talgam-Cohen are supported by the European Research Council (ERC) under the EU Horizon program (grant No.~101077862, project ALGOCONTRACT), by ISF grant No.~3331/24, and by a Google Research Scholar Award.
M.~Castiglioni is supported by the FAIR project PE0000013 and the EU Horizon project ELIAS (No.~101120237).
J.~Chen is supported by JST ERATO Grant No.~JPMJER2301, Japan. 
T.~Ezra is supported by the Harvard University Center of Mathematical Sciences and Applications. 
Y.~Li is supported by the NUS Start-up Grant.

\newpage
\bibliography{bibliography}
\bibliographystyle{ACM-Reference-Format}

\newpage
\appendix

\section{Generalized Scaling Lemma}\label{appx:scaling-lemma}
In this appendix, we present a generalized version of the scaling lemma developed in \cite{DuettingEFK23}. \citet{DuettingEFK23} proved that given an XOS function $f:2^N\rightarrow \R_{\geq 0}$ and set $T\subseteq A$, one can find a set $U\subseteq T$ with a desired value $\psi$ (up to a multiplicative and an additive constants), while not decreasing the marginals more than a constant factor. We next develop a new version of the scaling lemma in which the guarantee of the marginals is not with respect to the given set $T$ but with respect to an arbitrary set $S$ containing $T$. This modification is necessary for our analysis since we apply the scaling lemma to the set devised by our rounding scheme rather than directly use it on the set found by the \eqref{lp:lp1}.
\begin{lemma}[Generalized Scaling Lemma]
\label{lemma:scalinglemma-new}
Let $f : 2^{[n]} \to \mathbb{R}_{\geq 0}$ be an XOS function. Given inclusion minimal sets $T \subseteq S \subseteq [n]$,
and parameters $\delta \in (0, 1]$ and $0 \leq \Psi < f(T)$, Algorithm~\ref{alg:new-scaling} runs in polynomial time with value oracle access to $f$ and finds a set $U \subseteq T$ such that
\begin{align}\label{eq:scaling-first}
    (1-\delta)\Psi \leq f(U) \leq \Psi + \max_{i \in T} f(\{i\}), 
\end{align}
and
\begin{align}\label{eq:scaling-sec}
    f(i \mid U \setminus \{i\}) \geq \delta f(i \mid S \setminus \{i\}) \quad \text{for all } i \in U. 
\end{align}
\end{lemma}

We present Algorithm~\ref{alg:scalinglemma-new}, which essentially follows the approach of the scaling lemma but adjusts by removing elements based on the marginals of the superset $S$ instead of $T$.

\begin{algorithm}[H]
\caption{New Scaling Lemma for XOS\\
\textbf{Input:} XOS function $f$, inclusion minimal sets $T \subseteq S$, parameters $\delta \in (0, 1]$ and $0 \leq \Psi < f(T)$\\
\textbf{Output:} A set $U\subseteq T$ that satisfies \eqref{eq:scaling-first} and \eqref{eq:scaling-sec}}\label{alg:new-scaling}
\begin{algorithmic}[1]
\State Let $S_0=S, T_0=T$
\For{$s = 1, \dots, |T_0|$}
    \State Let $i_s = \arg\min_{i \in T_{s-1}} \frac{f(i \mid T_{s-1} \setminus \{i\})}{f(i \mid S_0 \setminus \{i\})}$
    \State Let $T_s = T_{s-1} \setminus \{i_s\}$
    \State Let $\delta_s = \frac{f(i_s \mid T_s)}{f(i_s \mid S_0 \setminus \{i_s\})}$
\EndFor
\State Let $j^{\star} = \min\{j \mid f(T_j) \leq \Psi\}$
\State Let $k^{\star} = \min\{k \mid f(T_k) \leq (1 - \delta) f(T_{j^{\star}-1})\}$
\State Let $s^{\star} = \arg\max_{t \in \{j^{\star}, \dots, k^{\star}\}} \delta_s$
\State \textbf{return} $U = T_{s^{\star}-1}$
\end{algorithmic}
\label{alg:scalinglemma-new}
\end{algorithm}

\begin{proof}
Observe that $f(T) = f(T_0) \geq \cdots \geq f(T_{|T_0|}) = 0.$ There must exist a smallest $0 \leq j \leq |T_0|$ such that $f(T_j) \leq \Psi$. This is $j^{\star}$. Furthermore, $j^{\star} \geq 1$ since $\Psi < f(T)$. Similarly, there must exist a smallest $j^{\star} - 1 \leq k \leq |T_0|$ such that $f(T_k) \leq (1 - \delta) f(T_{j^{\star} - 1}).$ This is $k^{\star}$. Since $f(T_{j^{\star} - 1}) > \Psi$ and $\delta >0$, it holds that $k^{\star} \geq j^{\star}$.

Now we show that $U=T_{s^{\star}-1}$ fulfills properties (1) and (2).
By definition, we have $T_s=T_{s-1}\setminus \{i_s\}$, therefore $f(i_{s} \mid T_{s})=f(T_{s-1})-f(T_s)$ which by taking a telescopic sum implies
$$
f(T_{j^{\star}-1})=f(T_{k^{\star}})+\sum_{s\in \{j^{\star},\hdots,k^{\star}\}}f(i_s\mid T_s).
$$
By rearranging and replacing $f(i_s\mid T_s)$ by $\delta_s f(i_s\mid S_0\setminus\{i_s\})$, we have
\begin{align*}
f(T_{j^{\star}-1})-f(T_{k^{\star}})=\sum_{s\in \{j^{\star},\hdots,k^{\star}\}} \delta_s f(i_s\mid S_0 \setminus \{i_s\})\leq  \sum_{s\in \{j^{\star},\hdots,k^{\star}\}}  \delta_{s^{\star}} f(i_s\mid S_0 \setminus \{i_s\}) \\
\leq \delta_{s^{\star}}f(\{i_{j^{\star}},\hdots, i_{k^{\star}}\})\leq  \delta_{s^{\star}}f(\{i_{j^{\star}},\hdots, i_{|T_0|}\})=\delta_{s^{\star}}\cdot f(T_{j^{\star}-1})    
\end{align*}
where the first inequality is by definition of $\delta_{s^{\star}}$, the second inequality is by Lemma~\ref{lemma:sum-of-marginals}, and the third inequality is by monotonicity. Therefore, $\delta_{s^{\star}}\geq \frac{f(T_{j^{\star}-1})-f(T_{k^{\star}})}{f(T_{j^{\star}-1})}.$

This implies for all $i\in U=T_{s^{\star}-1}$,
\begin{eqnarray*}
\frac{f(i\mid T_{s^{\star}-1}\setminus \{i\})}{f(i\mid S_0\setminus \{i\})}\geq \delta_{s^{\star}} \geq \frac{f(T_{j^{\star}-1})-f(T_{k^{\star}})}{f(T_{j^{\star}-1})} 
= 1-\frac{f(T_{k^{\star}})}{f(T_{j^{\star}-1})} \geq 1-\frac{(1-\delta)f(T_{j^{\star}-1})}{f(T_{j^{\star}-1})} 
\geq \delta,
\end{eqnarray*}
where the first inequality holds by the definitions of $\delta_{s^{\star}}$ and $i_{s^{\star}}$. 

To prove property \eqref{eq:scaling-first}, observe that $f(S_{s^{\star}-1})\geq f(S_{k^{\star}-1})\geq (1-\delta)f(S_{j^{\star}-1})>(1-\delta)\Psi,$
and that by sub-additivity $f(S_{s^{\star}-1})\leq f(S_{j^{\star}-1})\leq f(S_{j^{\star}})+\max_{i\in T}f(\{i\}),$ which is at most $\Psi+\max_{i\in T}f(\{i\})$ by the definition of $j^{\star}.$
\end{proof}









\section{Auxiliary Results from Previous Work}

\begin{lemma}[Lemma 3.4 in \cite{DuettingEFK23}]\label{lemma:lemma34-multi}
Fixing a project $j\in M$.   Any set $S$ that fulfills $f_j(S)>0$ and 
\begin{eqnarray*}
f_j(i\mid S\setminus \{i\})\geq \sqrt{2c_{ij} f_j(S)} && \forall i\in S,
\end{eqnarray*}
satisfies that $\rev_j(S)\geq f_j(S)/2$. 
\end{lemma}


\begin{lemma}[Lemma 2.1. \cite{DuettingEFK23}]\label{lemma:sum-of-marginals}
For any XOS function $f$ and any sets $S\subseteq T$, 
$$
\sum_{i\in S} f(i\mid T\setminus \{i\}) \leq f(S).
$$
\end{lemma}

\begin{lemma}[Lemma 3.3 \cite{DuettingEFK23}]\label{lemma:lemma33-multi} For every set $S\subseteq S^{\star}$ we have \begin{eqnarray*}
\sum_{i\in S}\sqrt{c_{i}}\leq \sqrt{f(S)}.
\end{eqnarray*}
\end{lemma}

We observe that the same result applies in our setting for $f_j$ and $S^{\star}_j$. Specifically, $S^{\star}_j$ is the optimal set for project $j$ when considering only the agents who are not assigned to other projects, i.e., the set $N \setminus \bigcup_{j' \neq j} S_{j'}$. Consequently, the same inequality holds for each project and its corresponding set from the optimal allocation.


\section{Submodular Functions with Value Queries}\label{appx:sub-mod}

In this section, we show how to extend our results when the functions $f_j(\cdot)$ are submodular using only value queries.
Indeed, if the functions are submodular, a demand oracle is not required as an approximate demand oracle can be implemented using only value queries.
This follows from two observations. First, the function $\min\{f_j(S),x\}$ is monotone submodular for any fixed $x$. Second, it is well known that the sum of a monotone submodular and a linear function can be approximate efficiently~\citep{sviridenko2017optimal,harshaw2019submodular}. Formally:

\begin{lemma}[\citep{sviridenko2017optimal,harshaw2019submodular}] \label{lemma:oracleSubmodular}
    Given a submodular function $f$, a threshold $x$, a set of prices $p_i$, and a value oracle for $f$, there exists a poly-time algorithm that finds a set $S$ such that
    \[\min\{f(S),x\}-\sum_{i\in S} p_i\ge \left(1-\frac{1}{e}\right) \min\{f(T),x\}- \sum_{i \in T} p_i \quad \forall T. \]
\end{lemma}

This approximate demand oracle can be used to replace the approximate demand oracle for capped function that we designed in Lemma~\ref{lemma:appx-demand} (without requiring demand queries).

Since the approximation provided in Lemma~\ref{lemma:oracleSubmodular} is slightly different from the one in Lemma~\ref{lemma:appx-demand}, we have to modify some components of the proof.
In particular, we start from a slightly different LP than \eqref{lp:lp1}, and prove a slightly different version of  Claim~\ref{claim:appx-separation}, Claim~\ref{claim:solve-lp1}, and Claim~\ref{claim:opt-sol-to-opt}. From there, the proof extends without any additional modifications.

We start from the following linear program:

\begin{align}
	\label{lp:lp1Submodular}\tag{LP4}
	\max & \sum_{j\in [m],x\in D_j,\scriptscriptstyle S\subseteq A_{j,x}} y_{j,x,\scriptscriptstyle S} \left(\left(1-\frac{1}{e}\right)\min\{f_j(S),x\}-\frac{\sqrt{x}}{2 \sqrt{2}}\sum_{i\in S}\sqrt{c_{ij}}\right)&\\
	\text{s.t.} & \sum_{x\in D_j,\scriptscriptstyle S\subseteq A_{j,x}} y_{j,x,\scriptscriptstyle S}\leq 1 &\quad \forall j\in [m],\nonumber\\
	&\sum_{j\in [m],x\in D_j, \scriptscriptstyle S \subseteq A_{j,x}\mid i\in S}y_{j,x,\scriptscriptstyle S}\leq 1  &\quad \forall i\in [n],\nonumber
\end{align}

and its dual

\begin{equation}\label{LP:lp2Submodular}\tag{LP5}
\begin{array}{ll@{}ll}
\text{min} & \sum_{j \in [m]} \alpha_j + \sum_{i \in [n]} \beta_i & \\
\text{s.t.} & \alpha_j + \sum_{i \in S} \beta_i \geq \left(1-\frac{1}{e}\right)\min\{f_j(S),x\} -\frac{\sqrt{x}}{2\sqrt{2}} \sum_{i \in S} \sqrt{c_{ij}} \quad \forall j \in [m], x \in D_j, S \subseteq A_{j,x},\\
& \alpha_j \geq 0, \quad \forall j \in [m],\\
& \beta_i \geq 0, \quad \forall i \in [n].
\end{array}.
\end{equation}

We can design the following approximate separation oracle for \eqref{LP:lp2Submodular}.
\begin{claim}[Approximate Separation Oracle]\label{claim:appx-separationSubmodular}
Given any assignment of the dual variables $ \alpha_j $ for all $ j \in M $ and $ \beta_i $ for all $ i \in N $ corresponding to \eqref{LP:lp2}, there exists a polynomial-time approximate separation oracle that either:  
\begin{enumerate}  
    \item Verifies that all constraints of \eqref{LP:lp2} are satisfied, i.e.,
   $$
    \alpha_j + \sum_{i \in S} \beta_i \geq \left(1-\frac{1}{e}\right) \min\{f_j(S), x\} - \sum_{i \in S} \frac{\sqrt{c_{ij}x}}{2\sqrt{2}}  
    \quad \forall j \in [m], x \in D_j, S \subseteq A_{j,x}.
    $$
    \item Otherwise, it identifies an \emph{approximately} violated constraint. Specifically,
    $$
    \exists j \in M, x \in D_j, S \subseteq A_{j,x}\; \text{such that} \; 
     \alpha_j + \sum_{i \in S} \beta_i <  \min\{f_j(S), x\} 
    - \sum_{i \in S} \frac{\sqrt{c_{ij}x}}{2\sqrt{2}}.
    $$    
\end{enumerate}  
\end{claim}

\begin{proof}
We construct the approximate demand oracle using the approximate capped demand in \cref{lemma:oracleSubmodular}.

Given an assignment $\{\alpha_j\}_{i\in M},\{\beta_i\}_{i\in N}$ for \eqref{LP:lp2Submodular}, we define the prices as in \Cref{claim:appx-separation} and compute an approximate demand of function $f_j$ with cap $x$ and prices $p_i$ using \cref{lemma:oracleSubmodular}.
Our approximate demand  yields a set $\bar{S} \subseteq A_{j,x}$ that satisfies 
\[\min\{f_j(\bar{S}), x\} - \sum_{i \in \bar{S}} p_i \geq \max_{S} \left(\left(1-\frac{1}{e}\right)\min\{f_j(S), x\} - \sum_{i \in S} p_i \right).\]

If $\alpha_j \geq   \min\{f_j(\bar{S}), x\} - \sum_{i \in \bar{S}} p_i$ for all $j\in M$ and $x\in D_j$,
all the constraints of \eqref{LP:lp2Submodular} are satisfied. This is because 
\begin{align*}
    \alpha_j  \geq  \min\{f_j(\bar{S}), x\} - \sum_{i \in \bar{S}} p_i\ge \left(1-\frac{1}{e}\right)\min\{f_j(S), x\} - \sum_{i \in S} p_i \quad \forall j\in M, x\in D_j, S \subseteq  A_{j,x}.
\end{align*}

Otherwise, there exists $j\in M$ and $x\in D_j$ such that $\left( \min\{f_j(\bar{S}), x\} - \sum_{i \in \bar{S}} p_i \right) > \alpha_j$. In this case, the constraint for $j$ and $x$ is approximately violated for $\bar{S}$. To see that, it is sufficient to replace the definition of the prices.
Since there is only a polynomial number of choices for $ j \in M$ and $x \in D_j$, this process terminates in polynomial time. 
\end{proof}

Let $P^{\star}_S = D^{\star}_S$ be the optimal solutions to the primal \eqref{lp:lp1Submodular} and the dual \eqref{LP:lp2Submodular}, respectively.
Then, running \Cref{alg:appx-demand} with the approximation oracle designed in \cref{claim:appx-separationSubmodular} and \eqref{lp:lp1Submodular} and \eqref{LP:lp2Submodular} as primal and dual problem, respectively, we get the following:

\begin{claim}\label{claim:solve-lp1Submodular} 
For every $\epsilon>0$, \cref{alg:solve-lp1} computes a feasible solution $y$ to \eqref{lp:lp1Submodular} in polynomial time that satisfies
\begin{enumerate}
   \item The number of strictly positive variables $y_{j,x,S}$ is polynomially bounded.\label{item:poly-manySubmodular}

    \item $\sum_{j \in M} \sum_{x \in D_j} \sum_{S \subseteq A_{j,x}}y_{j,x,S}  \left(\min\{f_j(S),x\} - \frac{\sqrt{x}}{2\sqrt{2}} \sum_{i \in S} \sqrt{c_{ij}} \right)\geq (P^{\star}_S-\epsilon).$\label{item:opt-boundSubmodular}
    
    \item For every $y_{j,x,S}>0$, $f_j(S)\leq (1+\delta) x$, $f_j(\{i\})\leq \delta x$ and $f_j(i\mid S_j\setminus \{i\})\geq \frac{\sqrt{xc_{ij}}}{2\sqrt{2}}$ $\forall i\in S.$\label{item:item-condSubmodular} 
\end{enumerate} 
\end{claim}

\begin{proof}

The proof proceeds similarly to the one of \Cref{claim:solve-lp1}.
During the ellipsoid method execution for \eqref{LP:lp2} with the additional constraint that the objective is less than $\gamma_-$, the algorithm identifies polynomially many \emph{approximately} violating constraints (or separating hyperplanes). These constraints correspond to tight constraints in the following program, which is similar to \eqref{LP:lp2Submodular} but expressed in the form of the approximate constraints.
\begin{equation*}
\begin{array}{ll@{}ll}
\text{min} & \sum_{j \in M} \alpha_j + \sum_{i \in N} \beta_i & \\
\text{s.t.} & \alpha_j \geq  \min\{f_j(S),x\} - \frac{\sqrt{x}}{2\sqrt{2}} \sum_{i \in S} (\sqrt{c_{ij}} + \beta_i )&& \forall j \in M, x \in D_j, S \subseteq A_{j,x},\\
& \alpha_j \geq 0, && \forall j \in M,\\
& \beta_i \geq 0, && \forall i \in N.
\end{array}
\end{equation*}

\noindent
Since there are only polynomially many tight constraints, the dual of this program contains only polynomially many non-slack variables. This results in the following ``small'' problem with polynomially many variables and constraints, which can be solved in polynomial time.
\begin{equation*}\label{LP:dual-dualSubmodular}\tag{LP6}
\begin{array}{ll@{}ll}
\text{max} &  \sum_{j \in M} \sum_{x \in D_j} \sum_{S \subseteq A_{j,x}} y_{j,x,S}  \left(\min\{f_j(S),x\} - \frac{\sqrt{x}}{2\sqrt{2}} \sum_{i \in S} \sqrt{c_{ij}} \right), \\
\text{s.t.} & \sum_{x\in D_j}\sum_{ S\subseteq A_{j,x}} y_{j,x, S}\leq 1 \quad \forall j\in M,\nonumber\\
&\sum_{j\in M}\sum_{x\in D_j}\sum_{  S \subseteq A_{j,x} \mid i\in S}y_{j,x, S}\leq 1  \quad \forall i\in N,\nonumber
\end{array}
\end{equation*}
By duality, since the dual is infeasible for $\gamma_-$, the reduced primal LP must have a feasible solution with an objective value of at least $\gamma_-$, which can be computed in polynomial time. This proves \eqref{item:poly-many} and \eqref{item:opt-bound}. 
\eqref{item:item-condSubmodular} can be proved similarly to
\Cref{claim:solve-lp1}.
\end{proof}

\begin{claim}\label{claim:opt-sol-to-optSubmodular}
For any $\epsilon<\delta \min_{j,i}f_j(\{i\})$, the solution computed by Algorithm~\ref{alg:solve-lp1} satisfies
$$\frac{1}{10}\sum_{j \in M}f_j(S^{\star \star}_j)\leq \sum_{j \in M} \sum_{x \in D_j} \sum_{S \subseteq A_{j,x}}y_{j,x,S}  {f_j(S)}.$$
\end{claim}

The constant in this case differs from that in the main result; however, this discrepancy does not change the result - it merely leads to a slightly smaller approximation constant.

\begin{proof}
First,
\begin{eqnarray*}
    \sum_{j \in M} \sum_{x \in D_j} \sum_{S \subseteq A_{j,x}}y_{j,x,S}  {f_j(S)}\geq\sum_{j \in M} \sum_{x \in D_j} \sum_{S \subseteq A_{j,x}}y_{j,x,S}  \left(\min\{f_j(S),x\} -  \sum_{i \in S}\frac{\sqrt{c_{ij} x}}{2\sqrt{2}}  \right),
\end{eqnarray*}
since $c_{ij}\geq 0$ and $x\geq 0$ for all $i\in N, j\in M$ and $x\in D_j$.
By \cref{claim:solve-lp1Submodular}, we have 
$$\sum_{j \in M} \sum_{x \in D_j} \sum_{S \subseteq A_{j,x}}y_{j,x,S}  {f_j(S)}\geq (P^{\star}-\epsilon).$$
As in \Cref{claim:opt-sol-to-opt}, we can show that $S^{\star \star}$ induces a feasible solution to \eqref{lp:lp1Submodular}. 
Hence, similarly:
\begin{align*}
P^{\star}&\geq \sum_{j \in M}\bigg(\left(1-\frac{1}{e}\right)\min\{f_j(S^{\star\star}_j),x'_j\}-\sum_{i\in S^{\star \star}_j}\frac{\sqrt{c_{ij}x'_j}}{2 \sqrt{2}}\bigg)\\
&\geq \sum_{j \in M}\bigg(\bigg(\left(1-\frac{1}{e}\right)f_j(S^{\star\star}_j)-\frac{1}{2}\sum_{i\in S^{\star \star}_j}\sqrt{c_{ij}f_j(S^{\star\star}_j)}\bigg)\\
&\geq\sum_{j \in M}\left(\frac{1}{2}-\frac{1}{e}\right)f_j(S^{\star\star}_j).
\end{align*}
The proof is completed following a similar argument to \Cref{claim:opt-sol-to-opt}. In particular, it holds
\begin{align*}
    \sum_{j \in M} \sum_{x \in D_j} \sum_{S \subseteq A_{j,x}}y_{j,x,S} f_j(S) 
    &\geq \sum_{j \in M} \sum_{x \in D_j} \sum_{S \subseteq A_{j,x}}y_{j,x,S}  \left(\min\{f_j(S),x\} - \frac{\sqrt{x}}{2\sqrt{2}} \sum_{i \in S} \sqrt{c_{ij}} \right)\\
    &\geq (P^{\star}_S-\epsilon)\\
    &\geq \left(\frac{1}{2}-\frac{1}{e}\right) \sum_{j \in M}f_j(S^{\star\star}_j)-\epsilon\\
    & \ge  \left(\frac{1}{2}-\frac{1}{e}-\delta \right)  \sum_{j \in M}f_j(S^{\star\star}_j).
\end{align*}
Since $\delta \leq \frac{1}{129}$, it holds  $\sum_{j \in M} \sum_{x \in D_j} \sum_{S \subseteq A_{j,x}}y_{j,x,S} f_j(S)\ge \frac{1}{10} \sum_{j \in M}f_j(S^{\star\star}_j)$.
This completes the proof. 
\end{proof}

\end{document}